\documentclass[fontsize=12pt,a4paper,headings=normal,
twoside=false,leqno,parskip=half-,abstract=true]{scrartcl}
\usepackage[english]{babel}
\usepackage[utf8]{inputenc}
\usepackage{hyperref}
\hypersetup{
 pdftitle={Chaos in subcritical HL cosmologies},
 pdfauthor={Phillipo Lappicy},
 colorlinks=true,
 linkcolor=blue,
 citecolor=blue,
 filecolor=blue,
 urlcolor=blue}

\usepackage{graphicx}
\usepackage[format=plain,labelfont=bf,font=small]{caption}
\usepackage{subfigure}
\usepackage{xcolor}
\usepackage[arrow, matrix, curve]{xy}
\usepackage{float}
\usepackage{caption}
\usepackage{tikz}
\usepackage{tikz-cd}
\usetikzlibrary{arrows, arrows.meta, calc, intersections, decorations.markings}
\tikzset{
decoration={markings,mark=at position 0.67 with {\arrow[thick,color=gray]{latex}}}
}

\usepackage{amsmath,amsthm}
\usepackage{amssymb} 
\usepackage{latexsym}

\usepackage{wasysym}

\DeclareOldFontCommand{\bf}{\normalfont\bfseries}{\mathbf}
\DeclareOldFontCommand{\it}{\normalfont\itshape}{\textit}

\usepackage{comment}

\newcommand{\mc}{\mathcal}
\newcommand{\mb}{\mathbb}

\newtheorem{thm}{Theorem}[section]
\newtheorem{lem}[thm]{Lemma}

\usepackage[notref,notcite,color,final 
]{showkeys}

\newcommand{\RR}{\mathbb R}

\newcommand{\BIIa}{\mathcal{B}_{N_1}}
\newcommand{\BIIb}{\mathcal{B}_{N_2}}
\newcommand{\BIIc}{\mathcal{B}_{N_3}}

\newcommand{\KC}{\mathrm{K}^{\ocircle}}

\newcommand{\KM}{\mathcal{K}}

\renewcommand{\S}{\cal S}

\newcommand{\Q}{{\bf \mathrm{Q}}}
\newcommand{\T}{{ \mathrm{T}}}

\begin{document}

\title{{\LARGE{Chaos in spatially homogeneous Ho\v{r}ava-Lifshitz subcritical cosmologies}}}

\author{
 \\
{~}\\
Phillipo Lappicy* and Victor H. Daniel*\\
\vspace{2cm}}

\date{ }
\maketitle
\thispagestyle{empty}

\vfill

$\ast$\\
ICMC, Universidade de S\~ao Paulo\\
Av. trabalhador são-carlense 400, 13566-590, São Carlos, SP, Brazil\\


\newpage
\pagestyle{plain}
\pagenumbering{arabic}
\setcounter{page}{1}

\begin{abstract}
\noindent We consider spatially homogeneous models in Ho\v{r}ava-Lifshitz (HL) gravity that perturbs General Relativity (GR) by a parameter $v\in (0,1)$ such that GR occurs at $v=1/2$. We prove that the induced Kasner map is chaotic for a broad class of modified HL gravity models, when $v\in (0,1/2)$, despite the fact that the Kasner map is multi-valued in such subcritical regime. 

\ 


\end{abstract}

\section{Introduction}\label{sec:intro}

\numberwithin{equation}{section}
\numberwithin{figure}{section}
\numberwithin{table}{section}

Ho{\v{r}}ava proposed a renormalizable, higher order derivative gravity theory that recovers general relativity (GR) in low energy but with improved high-energy behaviors, see \cite{hor09a,hor09b, HL_status_report}. This approach violates full spacetime diffeomorphism and introduce anisotropic scalings of space and time.
The deformation of the kinetics was firstly considered by DeWitt in \cite{DeWitt67}, whereas higher order derivatives in the potential was originally suggested by Lifshitz in \cite{Lifshitz41}.

More precisely, Ho{\v{r}}ava gravity is a gauge theory formulated in terms of a lapse $N$ and a shift vector $N^i$, which serve as Lagrange multipliers for the constraints in a Hamiltonian context, and a three-dimensional Riemannian metric $g_{ij}$ on the slices of the preferred foliation. We consider projectable theories, when the lapse depends only on time, $N\equiv N(t)$.
These objects arise from a 3+1 decomposition of a 4-metric according to,
\begin{equation}\label{genmetric}
\mathbf{g} = -N^2dt\otimes
dt + g_{ij}(dx^i + N^idt)\otimes (dx^j + N^jdt).
\end{equation}
In suitable units and scalings, the dynamics of Ho{\v{r}}ava vacuum gravity is governed by the
action
\begin{subequations}\label{action}
\begin{equation}
S = \int N\sqrt{ \det g_{ij}}({\cal T} - {\cal V}) dtd^3x,
\end{equation}
where ${\cal T}$ and ${\cal V}$ are given by
\begin{align}
{\cal T} &= K_{ij}K^{ij} - \lambda (K^k\!_k)^2,\label{kin}\\
{\cal V} &= {}^1 {\cal V}+{}^2 {\cal V}+{}^3 {\cal V}+{}^4 {\cal V}+{}^5 {\cal V}+{}^6 {\cal V}+\dots ,\\
&= k_1 R + k_2 R^2 + k_3 R^i\!_jR^j\!_i + k_4 R^i\!_jC^j\!_i + k_5 C^i\!_jC^j\!_i + k_6 R^3+ \dots\, .\label{calV}
\end{align}
\end{subequations}
Here $K_{ij}$ is the extrinsic curvature, $R$ and $R_{ij}$ are the scalar curvature
and Ricci tensor (of the spatial metric $g_{ij}$), respectively,
while $C_{ij}$ is the Cotton-York tensor~\cite{hor09b}, while the constants $\lambda, k_1,\dots ,k_6$ are real parameters. 
Each potential term ${}^i {\cal V}$, where $i=1,\dots,6$, is defined as the $i^{th}$ summand in \eqref{calV}.
Repeated indices are summed over according to Einstein's
summation convention.

Full spacetime diffeomorphism invariance in GR fixes $\lambda=1$
uniquely and set all parameters of ${\cal V}$ in~\eqref{calV} to
zero, except $k_1=-1$ (i.e., ${\cal V} = -R$), see~\cite{hor09a,hor09b}.
Thus GR is a special case among the Ho{\v{r}}ava models.
The introduction of $\lambda$ changes the scaling properties of the field
equations, as does the introduction of additional curvature terms. Since some of
the curvature terms have different scaling properties, sums of such terms in ${\cal V}$
result in that the field equations no longer are scale-invariant.

The classical Belinski, Khalatnikov and Lifshitz (BKL) picture suggests that generic singularities in GR are: (i) \textit{vacuum dominated}, (ii) \textit{local} and (iii) \textit{oscillatory}.
In this regard, vacuum spatially homogeneous cosmologies, the Bianchi models, are expected to play a key role in the dynamical asymptotic behaviour, see \cite{mis69a,bkl70,bkl82,Hobill94,Mixmaster,ugg13a,ugg13b}.
\textcolor{black}{Similarly, the Bianchi models in Ho{\v{r}}ava gravity are also expected to describe generic singularities, see \cite{Bakas10,Miso11,GiKam17,HellLappicyUggla}. 
In general, it is heuristically argued that there is an `asymptotically dominant' curvature term in \eqref{calV} toward the initial singularity, yielding a respective dominant Bianchi model, see \cite[Appendix A]{HellLappicyUggla}.}
%

\textcolor{black}{In Appendix \ref{app}, we deduce the} dominant Bianchi models in vacuum Ho{\v{r}}ava gravity:
\begin{subequations}\label{full:subs}
	\begin{align}
	\Sigma_+^\prime  &= 4v(1-\Sigma^2)\Sigma_+ + \, {\cal S}_+ \label{full:Sigma+}\\
	\Sigma_-^\prime   &= 4v(1-\Sigma^2)\Sigma_- + \, {\cal S}_- \label{full:Sigma-}\\
	N_1^\prime 			  &= -2(2v\Sigma^2 - 2\Sigma_+) \, N_1 \label{full:N1} \\
	N_2^\prime 			  &= -2(2v\Sigma^2 + \Sigma_+ + \sqrt{3}\Sigma_-)N_2 \label{full:N2} \\
	N_3^\prime 			  &= -2(2v\Sigma^2 + \Sigma_+ - \sqrt{3}\Sigma_- )N_3 \label{full:N3} 
	\end{align}
\end{subequations}
for some parameter $v \in [0,1]$, where the vector field is defined as follows
\begin{subequations}
	\begin{align}
	\Sigma^2  &	:= \Sigma_+^2 + \Sigma_-^2 ,\label{Sigma}\\
	\S_+ &  := 2 \left[ \left(N_3-N_2\right)^2 - N_1\left(2N_1-N_2-N_3\right) \right],\\
	\S_- &  := 2\sqrt{3} \left(N_3-N_2\right)\left(N_1-N_2-N_3\right).
	\end{align}
\end{subequations}
We denote ${}^\prime = d/d\tau$ the time derivative with respect to a time variable 
such that the singularity occurs as $\tau \rightarrow \infty$.
The evolution equations \eqref{full:subs} are bound to the constraint:
\begin{equation} \label{constraint}
1 =  \Sigma^2 + N_1^2 + N_2^2 + N_3^2 - 2N_1N_2 - 2N_2N_3 - 2N_3N_1,
\end{equation}
which restricts the phase-space $\RR^5$ to a four-dimensional invariant submanifold.


Equations \eqref{full:subs} describe vacuum spatially homogeneous models in GR when $v=1/2$, and thus $v\in (0,1)$ denotes perturbations of GR.
\textcolor{black}{Moreover, in \cite[Appendix A.2]{HellLappicyUggla}, they argue that the equations \eqref{full:subs} are expected to asymptotically describe the dynamics of each individual potentials ${}^i {\cal V}$, where $i=1, \dots ,6$, for the respective parameters,}
\begin{equation}\label{v's}
    {}^1v := \frac{1}{\sqrt{2(3\lambda-1)}}, \quad {}^{2}v = {}^{3}v = \frac{{}^1 v}{4}, \quad {}^4v = \frac{{}^1 v}{10}, \quad {}^{5}v = {}^{6}v = 0.
\end{equation}
For example, the exact evolution equations for vacuum spatially homogeneous $\lambda$-R models occur with parameter value ${}^1v$ defined in \eqref{v's}, whereas a Ho{\v{r}}ava model with quadratic (resp. cubic) potentials given by ${}^2 {\cal V}=k_2 R^2$ (resp. ${}^6 {\cal V}=k_6 R^3$) have dominant asymptotic evolution equation with parameter ${}^{2}v={}^1v/4$ (resp. ${}^{3}v=0$). 
\textcolor{black}{Instead of considering a particular ${}^iv$, we analyze all $v\in(0,1)$ which takes into account all possible dominant potentials in \eqref{calV}.}

The dynamics of the ODEs \eqref{full:subs} in GR, when $v=1/2$, has been extensively considered in the literature. A major achievement is the attractor theorem by H. Ringström, which states that the $\omega$-limit set of generic solutions of Bianchi type IX is contained in the space of solutions of Bianchi type I and II, see \cite{Ringstrom, heiugg09b}.
Therefore, it is expected that the concatenation of heteroclinic orbits of Bianchi type II and the induced map of Bianchi type I (the so-called \textit{Kasner map}) play a key role in the dynamics.
More rigorous results can be found, for example, in \cite{Beguin,Liebscher,Bernhard,Dutilleul}. A state~of~the~art overview is provided in \cite{Hobill94,Mixmaster}.

In the case that $v\neq 1/2$, some of the features of GR persist.
In particular, the Bianchi hierarchy of invariant sets remains valid depending on whether the variables $N_\alpha$ are nonzero. Bianchi type I consists of all $N_\alpha$ being zero, Bianchi type II consists of a single non-zero $N_\alpha$, Bianchi types $\mathrm{VI}_0,\mathrm{VII}_0$ consist of two non-zero $N_\alpha$, and Bianchi types $\mathrm{VIII},\mathrm{IX}$ consist of three non-zero $N_\alpha$.

Bianchi type I solutions are characterized by all $N_\alpha = 0,\alpha=1,2,3$.  The constraint \eqref{constraint} reduces the phase space to the  \textit{Kasner circle} of equilibria:
\begin{align} \label{defKC}
\KC := \left\{ \left(\Sigma_+,\Sigma_-, 0,0,0\right) \in \RR^5 \,\, |  \,\,
\Sigma_+^2 + \Sigma_-^2 = 1 \right\}.
\end{align}
There are three special points in $\KC$ corresponding to the Taub representation of Minkowski spacetime in GR. They are therefore called the \emph{Taub points} and given by
\begin{align}\label{Taubpoints}
\T_1 := \left(-1,0\right), \hspace{1.2cm}
\T_2 := \left(\dfrac{1}{2},\frac{\sqrt 3}{2}\right), \hspace{1.2cm} \T_3 := \left(\dfrac{1}{2},-\dfrac{\sqrt 3}{2}\right).
\end{align}
Note the existence of the Kasner circle is independent on the parameter $v\in [0,1]$. Its stability, however, depends strongly on the parameter and this affects the type II solutions, which will be described next. 

In general, linearization of equation~\eqref{full:subs} at $\mathrm{K}^{\ocircle}$ results in $N'_1 = 2(\Sigma_+|_{\mathrm{K}^{\ocircle}}-v){N}_1$, and thereby the stability behaviour of $N_{1}$ changes when $\Sigma_+|_{\mathrm{K}^{\ocircle}} = v$. We define the \emph{unstable Kasner arc}, denoted by $\mathrm{int}(A_{1})$, to be the points in $\mathrm{K}^{\ocircle}$ that are unstable in the $N_1$ variable, i.e., when $\Sigma_+ > v$. The closure of $\mathrm{int}(A_{1})$ is denoted by
$A_1$ and is given by
\begin{equation}\label{A_1}
A_1:= \left\{ (\Sigma_+,\Sigma_-,0,0,0)\in \mathrm{K}^{\ocircle} \text{ $|$ }
\Sigma_+ \geq v\right\}.
\end{equation}
Note that the arcs $A_1$ are symmetric portions of
$\mathrm{K}^{\ocircle}$ with points $\mathrm{Q}_1 := -\mathrm{T}_1$ in the middle.
The boundary set $\partial A_1$ consists of two fixed points, which coincide with the Taub points $\mathrm{T}_2$ and $\mathrm{T}_3$ when $v=1/2$, but unfolds each Taub point into two non-hyperbolic fixed points when $v\neq 1/2$.
Equivariance yields the arcs $A_2,A_3$, where the respective variables $N_2,N_3$ are unstable, and leads to similar statements for $\mathrm{T}_2$ and $\mathrm{T}_3$.
For $v<1/2$, we define 
\begin{equation}
    \mathbf{A_{\alpha\beta}}:= A_\alpha\cap A_\beta, \quad \text{ and } \quad \mathbf{A}:= \underset{\alpha\beta}{\bigcup} \, \mathbf{A_{\alpha\beta}},     
\end{equation}
where $(\alpha\beta\gamma)$ is a permutation of $(123)$.

Bianchi type II solutions consists of three disjoint hemispheres with a common boundary: the Kasner circle.
More precisely, it is the set of solutions where two $N$-variables are zero and one is nonzero, i.e. $\BIIa \cup \BIIb \cup \BIIc$, where  
\begin{align}\label{BII_N_1}
\BIIa :=\left\{ \left( \Sigma_+,\Sigma_-, N_1,0,0\right) \in \RR^5 \,\, \Big|  \begin{array}{c}
N_1 > 0  \\
{N_1}^2 = 1-\Sigma^2
\end{array} \right\},
\end{align}
and $\BIIb, \BIIc$ are obtained by equivariance with non-zero $N_2,N_3$, respectively.
Solutions of \eqref{full:subs} in the hemisphere $\BIIa$ are heteroclinic orbits between two Kasner equilibria with $\alpha$-limit sets in $\text{int}(A_1)$ and $\omega$-limit in the complement $A_1^c:=\mathrm{K}^{\ocircle} \backslash A_1$, as depicted in Figure \ref{KC_3d_map}. Similarly for $\BIIb, \BIIc$. For more details, see \cite{HellLappicyUggla}.
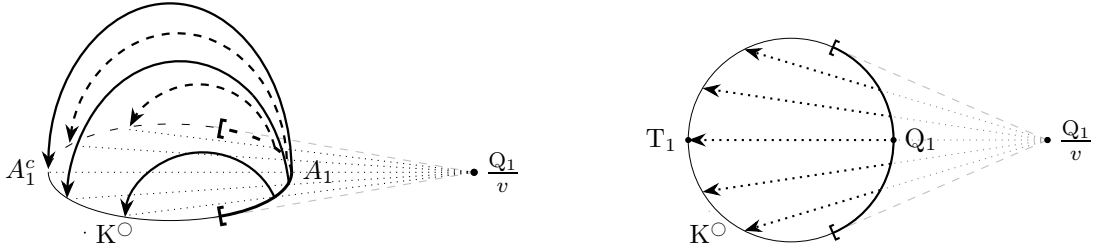
\begin{figure}[H]
	\centering
	\begin{tikzpicture}[scale=1.6]

 	\draw[lightgray,dashed,-] (2.5,0) -- (0.3,0.39);
 	\draw[lightgray,dashed,-] (2.5,0) -- (0.3,-0.39);

 	\draw [line width = 0.1mm] (-1,0) arc (180:360:1cm and 0.4cm);
 	\draw [line width = 0.1mm, loosely dashed] (1,0) arc (0:180:1cm and 0.4cm);
 	
 	\draw [-Bracket, line width = 0.4mm] (1,0) arc (0:-66.4218:1cm and 0.4cm);
 	\draw [-Bracket, line width = 0.4mm, loosely dashed] (1,0) arc (0:66.4218:1cm and 0.4cm);
 	\draw [black] (1,0) circle (0.01pt) node[anchor=west] {\footnotesize $A_1$};
  	\draw [black] (-1,0) circle (0.01pt) node[anchor=east] {\footnotesize $A_1^c$};

 	\filldraw [black] (-0.7,-0.5) circle (0.1pt) node[anchor=west] {\footnotesize $\KC$};

 	\draw [ -Stealth, line width = 0.3mm] (1,0) arc (0:179:1cm and 1.4cm);

	\draw [ -Stealth, line width = 0.3mm]  (0.86,-0.21) arc (22:174:0.64cm and 0.6cm);
	\draw [ -Stealth, line width = 0.3mm, dashed]  (0.86,0.21) arc (15:150:0.65cm and 0.7cm);

    \draw [line width = 0.3mm]  (0.68, 0.64) arc (48:8:0.93cm and 1.23cm);
    \draw [-Stealth,line width = 0.3mm]  (0.68, 0.64) arc (50:177:0.93cm and 1.2cm);
    \draw [-Stealth, line width = 0.3mm, dashed]  (0.96, 0.11) arc (3:170:0.90cm and 1.1cm);
 	
 	\draw[dotted] (2.5,0) -- (-1,0);
 	\draw[dotted] (2.5,0) -- (-0.8,0.22);
 	\draw[dotted] (2.5,0) -- (-0.8,-0.22);
	\draw[dotted] (2.5,0) -- (-0.35,0.356);
 	\draw[dotted] (2.5,0) -- (-0.35,-0.356);


%
%
%
%
%
%
%
%

 	\filldraw [black] (2.5,0) circle (0.7pt);
 	\draw [black] (2.5,0) circle (0.7pt) node[anchor=west] {$\frac{\Q_1}{v}$};
 	

	\end{tikzpicture}
	\qquad\quad
	\begin{tikzpicture}[scale=1.35]


	\draw[lightgray,dashed, -] (2.5,0) -- (0.41,-0.91);
	\draw[lightgray,dashed, -] (2.5,0) -- (0.41, 0.91);

	\draw [very thin] (1,0) arc (0:360:1cm and 1cm);	 	
	
	\draw [thick,Bracket-Bracket] (0.4,-0.9165) arc (-66.4218:66.4218:1cm and 1cm);
	
	\filldraw [black] (-1,0) circle (0.7pt) node[anchor=east] {\footnotesize $\T_1$};
	\filldraw [black] (1,0) circle (0.7pt) node[anchor=west] {\footnotesize $\Q_1$};
	
	\filldraw [black] (0.4,0.9165) circle (0pt);
	\filldraw [black] (0.4,-0.9165) circle (0pt);


\foreach \s in {0.15,0.34,0.5,0.66,0.85}
{\draw[gray,dotted] (2.5,0) -- ({cos(42-\s*84)},{sin(42-\s*84)});}
\draw[-Stealth,dotted,thick] ({cos(42-0.15*84)},{sin(42-0.15*84)}) -- ({cos(90+0.15*180)},{sin(90+0.15*180)});
\draw[-Stealth,dotted,thick] ({cos(42-0.5*84)},{sin(42-0.5*84)}) -- ({cos(90+0.5*180)},{sin(90+0.5*180)});
\draw[-Stealth,dotted,thick] ({cos(42-0.85*84)},{sin(42-0.85*84)}) -- ({cos(90+0.85*180)},{sin(90+0.85*180)});

\draw[-Stealth,dotted,thick] ({cos(42-0.34*84)},{sin(42-0.34*84)}) -- ({cos(90+0.33*180)},{sin(90+0.33*180)});
\draw[-Stealth,dotted,thick] ({cos(42-0.66*84)},{sin(42-0.66*84)}) -- ({cos(90+0.67*180)},{sin(90+0.67*180)});
%
%
%
%
	
	%
	\filldraw [black] (-0.8,-0.7) circle (0pt) node[anchor=north] {\footnotesize $\KC$}; 	
	
	\filldraw [black] (2.5,0) circle (0.7pt);
	\draw [black] (2.5,0) circle (0.7pt) node[anchor=west] {$\frac{\Q_1}{v}$};

	\end{tikzpicture}
	\caption {Left: Bianchi type II solutions are heteroclinic orbits in the hemisphere $\BIIa$ with $\alpha$-limits within int$(A_1)$ and $\omega$-limits in $A_1^c$. Right: Projection of the Bianchi type II solutions into the $\Sigma$-plane.}\label{KC_3d_map}
\end{figure}

Thus, the type II solutions in $\BIIa$ induce a map from the $\alpha$-limit set to the $\omega$-limit set of heteroclinics in $\BIIa$, denoted by $\KM_1:A_1 \longrightarrow \overline{A_1^c}$. In particular, we can parametrize $A_1$ by polar coordinates $(\Sigma_+,\Sigma_-)=(\cos(\varphi),\sin(\varphi))$, which implies that the angle variable $\varphi\in A_1=[-\arccos v, \arccos v]$.
In this coordinates, the map $\KM_1$ is given by
\begin{equation}\label{KM_angle}
\KM_1(\varphi) :=
\pi - 2 \arctan\left(\left(\frac{1+v}{1-v} \right) \tan\left(\frac{\varphi}{2}\right)\right),
\end{equation}
with derivative given by $|D\KM_1(\varphi)| = g(\varphi) := \frac{1-v^2}{ v^2 + 1 -2\cos(\varphi) v} \geq 1.$
%
%
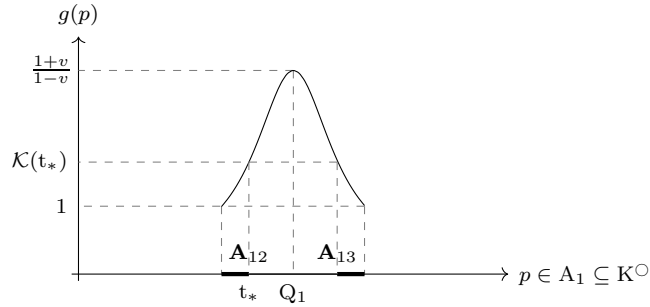
\begin{figure}[H]
\centering
    \begin{tikzpicture}[scale=0.9]
        \draw[->] (0,-0.1) -- (0,3.5)node[anchor=south] {\scriptsize{$g(p)$}};
        \draw[->] (-0.1,0) -- (6.28,0)node[anchor=west] {\scriptsize{$p\in \mathrm{A}_1 \subseteq \mathrm{K}^{\ocircle}$}};

        \draw[color=gray,dashed] (3.14,3) -- (0,3); \node at (-0.4,3) {\scriptsize{$\frac{1+v}{1-v}$}};
        \draw[color=gray,dashed] (3.14,3) -- (3.14,0) node[anchor=north] {\color{black} \scriptsize{$\mathrm{Q}_1$}};

        \draw[color=gray,dashed] (2.09,1) -- (2.09,0);

        \draw[color=gray,dashed] (4.18,1) -- (0,1) node[anchor=east] {\color{black}\scriptsize{ $1$}};
        \draw[color=gray,dashed] (4.18,1) -- (4.18,0);

        \draw[-] (2.09,0) -- (4.18,0);

        \draw [domain=2.09:4.18,variable=\t,smooth] plot ({\t},{(1-(1/2)^2)/(1+(1/2)^2+2*(1/2)*cos(\t r))});
        
        \draw[ultra thick] (2.09,0) -- (2.49,0) node[anchor=south] {\color{black} \scriptsize{$\mathbf{A}_{12}$}};
        \draw[ultra thick] (4.18,0) -- (3.78,0) node[anchor=south] {\color{black} \scriptsize{$\mathbf{A}_{13}$}};
        
        \draw[color=gray,dashed] (3.78,1.65) -- (0,1.65) node[anchor=east] {\color{black} \scriptsize{$\KM (\mathrm{t}_*)$}};
        \draw[color=gray,dashed] (2.49,1.65) -- (2.49,0) node[anchor=north] {\color{black} \scriptsize{$\mathrm{t}_{*}$}};
        \draw[color=gray,dashed] (3.78,1.65) -- (3.78,0);
    \end{tikzpicture}
\caption{The function $g(p)$ for $p\in A_1$. 
For $v\in (0,1/2)$, note that any open set $U\subseteq A_1\setminus (\mathbf{A}_{12}\cup\mathbf{A}_{13})$ satisfies $|\KM_1 (U) |\geq |D\KM_1 (\mathrm{t}_*)|\, |U|$ such that $|D\KM (\mathrm{t}_*)|>1$ and $\mathrm{t}_*$ is the unique point in $\partial \mathbf{A}_{12} \cap \mathrm{int}(A_1)$. }\label{fig:plotg}
\end{figure}
Analogous constructions in $\BIIb$ and $\BIIc$ yield maps $\KM_2$ and $\KM_3$,
which altogether
induce a discrete map of the circle 
called the \emph{Kasner map} $\KM:\KC \longrightarrow \KC $. 
See Figure \ref{FIG:KASNERMAPS}. 
Note that iterations of $\KM$ represent a sequence of heteroclinic Bianchi type II trajectories.

The map $\KM$ depends on $v$ according to three qualitatively different regimes, as in Figure \ref{FIG:KASNERMAPS}.
For the critical case $v=1/2$ associated with GR, the Kasner circle $\KC$ is inscribed in the triangle formed by the points $\Q_1,\Q_2,\Q_3$ and the map $\KM$ is known to be generically chaotic, see~\cite{mis69a,bkl70,Chernoff83,khaetal85,Cornish97,ugg13a,ugg13b} and references therein. For the supercritical case, $v\in (1/2,1)$, the three arcs are pairwise disjoint and their union does not cover $\KC$; its complement is a set consisting of stable fixed points of $\KM$ and thus it is called the \textit{stable set} $S$. The map $\KM$ is well-defined 
and it was shown that the set of points that never reaches the set $S$ is Cantor set of measure zero, 
in which the map $\KM$ is chaotic. 
See \cite{HellLappicyUggla}. 
For $v\in(0,1/2)$, the map $\KM$ is multi-valued for points in the overlap regions $A_\alpha\cap A_\beta$ and uniquely determined for the remaining points. Moreover, certain choices of the multi-valued map is discontinuous at $\partial (A_\alpha\cap A_\beta)$, but even when lateral limits exist, the derivative is 1 which makes the map weakly hyperbolic. For this reason, little is known in this case and the map was conjectured to be chaotic in \cite{HellLappicyUggla}.
\begin{figure}[H]
\minipage[b]{0.37\textwidth}\centering
\begin{subfigure}\centering
    \begin{tikzpicture}[scale=1.1]

    \draw [line width=0.1pt,domain=0:6.28,variable=\t,smooth] plot ({sin(\t r)},{cos(\t r)});

    \draw [very thick, domain=-0.29:0.29,variable=\t,smooth] plot ({0.975*sin(\t r)},{0.975*cos(\t r)});
    \draw [rotate=120,very thick, domain=-0.29:0.29,variable=\t,smooth] plot ({0.975*sin(\t r)},{0.975*cos(\t r)});
    \draw [rotate=-120,very thick, domain=-0.29:0.29,variable=\t,smooth] plot ({0.975*sin(\t r)},{0.975*cos(\t r)});

    \draw[color=gray,dashed] (0,-2.85) -- (0.98,-0.29);
    \draw[color=gray,dashed] (0,-2.85) -- (-0.98,-0.29);

    \draw[color=gray,rotate=120,dashed] (0,-2.85) -- (0.98,-0.29);
    \draw[color=gray,rotate=120,dashed] (0,-2.85) -- (-0.98,-0.29);

    \draw[color=gray,rotate=240,dashed] (0,-2.85) -- (0.98,-0.29);
    \draw[color=gray,rotate=240,dashed] (0,-2.85) -- (-0.98,-0.29);

    \filldraw[gray] (0,-2.85) circle (0.1pt);

    \draw[rotate=-120,color=gray,dotted] (-0.5,0.85) -- (0,-2.85);
    \draw[rotate=-120] (-0.5,0.85);
    \draw[rotate=-120] (-0.3,-0.95);

    \draw[rotate=-120,white,ultra thick] (-0.5,0.85) -- (-0.26,-0.95);
    \draw[rotate=-120,dotted, thick, postaction={decorate}] (-0.26,-0.95) -- (-0.5,0.85);




    \draw (0,-2.85) circle (0.1pt) node[anchor=north] {$\frac{\mathrm{Q}_1}{v}$};
    \draw[rotate=240] (0,-2.85) circle (0.1pt) node[anchor=east] {$\frac{\mathrm{Q}_2}{v}$};
    \draw[rotate=120]  (0,-2.85) circle (0.1pt) node[anchor=west] {$\frac{\mathrm{Q}_3}{v}$};

    \draw (0,-0.95) -- (0,-1.05) node[anchor= north] {\scriptsize{$\mathrm{Q}_1$}};
    \draw[rotate=120] (0,-0.95) -- (0,-1.05) node[anchor= west] {\scriptsize{$\mathrm{Q}_3$}};
    \draw[rotate=-120] (0,-0.95) -- (0,-1.05) node[anchor= east] {\scriptsize{$\mathrm{Q}_2$}};

    \node at (0,1.2) {\scriptsize{$A_2\cap A_3$}};
    \node at (-1.4,-0.57) {\scriptsize{$A_1\cap A_2$}};
    \node at (1.4,-0.57) {\scriptsize{$A_1\cap A_3$}};
\end{tikzpicture}
    \addtocounter{subfigure}{-1}\captionof{subfigure}{\footnotesize{Subcritical: $v\in (0,1/2)$. }}
\end{subfigure}
\endminipage\hfill
\minipage[b]{0.31\textwidth}\centering

\hspace{1cm}
\begin{subfigure}\centering
    \begin{tikzpicture}[scale=1.1]
    \draw[white] (0,-2.85) circle (0.1pt) node[anchor=north] {$\frac{\mathrm{Q}_1}{v}$};
    
    \draw [line width=0.1pt,domain=0:6.28,variable=\t,smooth] plot ({sin(\t r)},{cos(\t r)});

    \draw[color=gray,dashed,-] (-1.75,1) -- (1.75,1);
    \draw[color=gray,rotate=120,dashed,-] (-1.75,1) -- (1.75,1);
    \draw[color=gray,rotate=240,dashed,-] (-1.75,1) -- (1.75,1);

    \draw[color=gray,dotted] (0,1) -- (0,-1.99);
    \filldraw (0,-2) circle (0.1pt);

    \draw[white,ultra thick] (0,1) -- (0,-1);
    \draw[dotted, thick, postaction={decorate}] (0,-1) -- (0,1);


    \draw (0,-2) circle (0.1pt);
    \draw[rotate=240] (0,-2) circle (0.1pt);
    \draw[rotate=120]  (0,-2) circle (0.1pt);

    \filldraw [black] (0,1) circle (1.25pt) node[anchor= south] {\scriptsize{$\mathrm{T}_1$}};
    \filldraw [black] (0.88,-0.49) circle (1.25pt) node[anchor= north west] {\scriptsize{$\mathrm{T}_2$}};
    \filldraw [black] (-0.88,-0.49) circle (1.25pt)node[anchor= north east] {\scriptsize{$\mathrm{T}_3$}};

    \draw (0,-0.95) -- (0,-1.05);
    \draw[rotate=120] (0,-0.95) -- (0,-1.05);
    \draw[rotate=-120] (0,-0.95) -- (0,-1.05);
\end{tikzpicture}
    \addtocounter{subfigure}{-1}\captionof{subfigure}{\footnotesize{Critical: $v=1/2$. }}
\end{subfigure}
\endminipage\hfill
\minipage[b]{0.31\textwidth}\centering
\begin{subfigure}\centering
    \begin{tikzpicture}[scale=1.1]
    \draw[white] (0,-2.85) circle (0.1pt) node[anchor=north] {$\frac{\mathrm{Q}_1}{v}$};

    \draw [line width=0.1pt,domain=0:6.28,variable=\t,smooth] plot ({sin(\t r)},{cos(\t r)});

    \draw [ultra thick, dotted, white, domain=-0.26:0.26,variable=\t,smooth] plot ({sin(\t r)},{cos(\t r)});
    \draw [ultra thick, dotted, white, domain=1.83:2.35,variable=\t,smooth] plot ({sin(\t r)},{cos(\t r)});
    \draw [ultra thick, dotted, white, domain=3.87:4.39,variable=\t,smooth] plot ({sin(\t r)},{cos(\t r)});

    \draw[color=gray,dashed,-] (0,-1.35) -- (0.685,-0.75);
    \draw[color=gray,dashed,-] (0,-1.35) -- (-0.685,-0.75);

    \draw[color=gray,rotate=120,dashed,-] (0,-1.35) -- (0.685,-0.75);
    \draw[color=gray,rotate=120,dashed,-] (0,-1.35) -- (-0.685,-0.75);

    \draw[color=gray,rotate=-120,dashed,-] (0,-1.35) -- (0.685,-0.75);
    \draw[color=gray,rotate=-120,dashed,-] (0,-1.35) -- (-0.685,-0.75);

    \draw[rotate=120,color=gray,dotted] (0,-1.35) -- (-0.9,0.5);

    \draw[rotate=120,white,ultra thick] (-0.2,-0.95) -- (-0.88,0.45);
    \draw[rotate=120,dotted, thick, postaction={decorate}] (-0.2,-0.95) -- (-0.88,0.45);



    \draw (0,-1.35) circle (0.1pt);
    \draw[rotate=240] (0,-1.35) circle (0.1pt);
    \draw[rotate=120]  (0,-1.35) circle (0.1pt);


    \draw (0,-0.95) -- (0,-1.05); 
    \draw[rotate=120] (0,-0.95) -- (0,-1.05);
    \draw[rotate=-120] (0,-0.95) -- (0,-1.05);

    \node at (0,1.2) {\scriptsize{$S$}};
    \node at (-1.1,-0.57) {\scriptsize{$S$}};
    \node at (1.1,-0.57) {\scriptsize{$S$}};

\end{tikzpicture}
    \addtocounter{subfigure}{-1}\captionof{subfigure}{\footnotesize{Supercritical: $v\in (1/2,1)$. }}
\end{subfigure}
\endminipage
\captionof{figure}{The geometric description of the Kasner map $\KM$ in the three different regimes. The subcritical case, $v\in (0,1/2)$, each pair of arcs $A_1,A_2,A_3$ intersects at (bold) sets, where a point has two distinct trajectories and $\KM$ is multi-valued. The critical case, $v=1/2$, corresponds with GR and each pair of arcs $A_1,A_2,A_3$ only intersect at their boundary: the Taub points $\mathrm{T}_1,\mathrm{T}_2,\mathrm{T}_3$. The supercritical case, $v\in (1/2,1)$, the complement of the arcs $A_1,A_2,A_3$ is the (dashed) stable set $S$ that consists of fixed points of $\KM$.
}\label{FIG:KASNERMAPS}
\end{figure}
Our goal is to prove that the Kasner map is chaotic for $v\in (0,1/2)$, which thereby completes the program in  \cite{HellLappicyUggla}. As a consequence, we rigorously show that GR is a critical case associated with a bifurcation where chaos becomes generic\textcolor{black}{, i.e., GR is a specific model among Ho\v{r}ava gravity where a qualitative change on the structure of solutions occurs.} Even though irregular and non-hyperbolic maps impose difficulties in proving that a given map is chaotic, 
see \cite{Lima19}, it is the fact that the Kasner map is multi-valued that complicates the analysis. 

For $v\in (0,1/2)$, we extend the Kasner maps which are defined only in each arc $A_\alpha$ to the whole circle, i.e. we define the maps $\mathcal{K}_{\mu \nu \zeta}:\KC\rightarrow \KC$ given by
\begin{equation}\label{Kmnzdef}
    \mathcal{K}_{\mu \nu \zeta}(p):=\begin{cases}
    \mathcal{K}_\alpha(p)\text{, if }p\in A_\alpha\setminus \mathbf{A},\\
    \mathcal{K}_{\mu}(p)\text{, if }p\in A_1\cap A_2,\\
    \mathcal{K}_{\nu}(p)\text{, if }p\in A_2\cap A_3,\\
    \mathcal{K}_{\zeta}(p)\text{, if }p\in A_1\cap A_3,
\end{cases}
\end{equation}
for each $\mu\in \{1,2\}, \nu\in\{2,3\}, \zeta\in\{1,3\}$. 
%
Each map $\mathcal{K}_{\mu \nu \zeta}$ in \eqref{Kmnzdef} selects a particular map in the multi-valued region. To account for the full multi-valued dynamics, we have to consider all possible words $\mu\nu\zeta$ for the iterations. Hence, define the space of symbols
\begin{equation} 
    \Sigma:=\{\omega=(\mu_n,\nu_n,\zeta_n)_{n\ge 1}:\mu_n\in \{1,2\},\nu_n\in \{2,3\},\zeta_n\in\{1,3\} \text{ for all } n\ge 1\}.
\end{equation}
Thus for any $\omega=(\omega_n)_{n\ge 1} \in \Sigma$, we can define the following iterates for $p\in \KC$,
\begin{equation}
    \mathcal{K}_\omega^n(p):=\mathcal{K}_{\omega_n}\circ \dots \circ \mathcal{K}_{\omega_1}(p), 
\end{equation}
where $n\ge 1$ and $\mathcal{K}_{\omega}^0:=Id_{\KC}$. 
This describes the Kasner map in a skew-product fashion:
\begin{align} \label{defKM}
\KM: \Sigma\times \KC &\to \Sigma\times \KC \\
(\omega,p)&\mapsto (\sigma(\omega),\mathcal{K}_{\omega_1}(p)),
\end{align}
where $\sigma(\omega)=\sigma(\omega_1\omega_2\omega_3\ldots)=\omega_2\omega_3\ldots$ is the shift-map.

In order to describe the chaotic aspects of our skew-product Kasner map, we define a suitable notion of chaos. We say that the Kasner map \emph{realizes chaos} if:
\begin{itemize}
    \item[(i)] The Kasner map \emph{realizes sensitivity to initial conditions}, i.e., if there exists $\delta>0$ such that for every $p\in \KC$ and every neighborhood $U$ of $p$, there are $q\in U\setminus \{p\},\ n\in \mathbb{N}_0$ and $\omega,\omega^*\in \Sigma$ such that $d(\mathcal{K}_\omega^n(p),\mathcal{K}_{\omega^*}^n(q))\ge \delta$.
    \item[(ii)] The Kasner map \emph{realizes topological transitivity}, i.e., for any open sets $U,V\subset \KC$ there are $n\in \mathbb{N},\omega \in \Sigma$ such that $\mathcal{K}_\omega^n(U)\cap V\neq \emptyset$.
    \item[(iii)] The Kasner map \emph{realizes dense periodic orbits}, i.e., for every $p\in \KC$ and every neighborhood $U\subset \KC$ of $p$ there are $\omega \in \Sigma,q\in U,n\in \mathbb{N}$ such that $\mathcal{K}_\omega^n(q)=q$.
\end{itemize}
Note that the map in \eqref{Kmnzdef} can be appropriately extended for $v\in \left[1/2,1\right)$, for which the above definition coincides with Devaney's chaos, see \cite{Devaney}. 
Other definitions of chaos in such a multi-valued setting are possible, such as requiring generic $\omega\in \Sigma$. However, it renders a more difficult analysis to obtain rigorous results.
Next, we provide our main result.

\begin{thm}
\label{mainthm}
The (multi-valued) Kasner map realizes chaos for all $v\in (0,1/2)$.
\end{thm}
{
\color{black}
Sensitivity to initial data has been observed numerically for some HL models in \cite{Miso11}. 
Thus, Theorem~\ref{mainthm} goes beyond rigorously confirming these predictions by means of the Kasner map, as we establish the occurrence of chaos (via topological transitivity and dense periodic orbits).
Moreover, similar to GR, it is heuristically argued in \cite[Appendix A]{HellLappicyUggla} that the generic dynamics of Bianchi type IX models in HL gravity should be described by the appropriate Kasner map. See also \cite{Bakas10,GiKam17}. Thus, not only generic Bianchi type IX solutions in HL are expected to evolve apart (due to sensitivity), but also display a sequence of Kasner-like states by revisiting arbitrary neighborhoods of the Kasner circle (due to topological transitivity). Density of periodic orbits of the Kasner map indicates that generic Bianchi type IX solutions should shadow concatenations of heteroclinic chains of Bianchi type II solutions, for example within a stable manifold, see \cite{Liebscher}.

The remaining of the paper is organized as follows. In Section \ref{sec:pf}, we prove Theorem \ref{mainthm}. In Section \ref{sec:disc} we discuss the main features of our proof, some consequences and future directions to be pursued. In Appendix \ref{app}, we derive the dominant ODE model \eqref{full:subs}. 
}












\section{Proof}\label{sec:pf}


The proof of our main result relies on the following Lemma.

\begin{lem}\label{lem:intOVERLAP}
Consider a connected open set $U\subset \KC$.
\begin{enumerate}
    \item[(i)] For any $\omega\in \Sigma$, there is an $N\in \mathbb{N}_0$ such that $\mc{K}_{\omega}^N(U)$ is connected and $\mc{K}^N_{\omega}(U)\cap \mathbf{A}\neq \emptyset$. 
    \item[(ii)] There are $\omega \in \Sigma, M\in \mb{N}_0$ such that $\KM^M_\omega(U) \supseteq A_\alpha \cap A_\beta$ for some $\alpha, \beta \in \{1,2,3\}, \alpha\neq \beta$.
\end{enumerate}
\end{lem}
\begin{proof}
    %
    %
    %
    To prove $(i)$, suppose towards a contradiction that $\mathcal{K}^n_{\omega}(U)\cap \mathbf{A}=\emptyset$ for every $n\ge 0$. Note that for each $n\ge 0$ there exists $\alpha_n\in \{1,2,3\}$ such that $\mathcal{K}^n_{\omega}(U)$ is connected and $\mathcal{K}^n_{\omega}(U)\subset A_{\alpha_n}\setminus \mathbf{A}$. 
    Moreover, since the Kasner map is strictly expanding in $A_{\alpha_n}\setminus \mathbf{A}$, according to Figure \ref{fig:plotg}, we obtain that
    \begin{equation}
        |\mathcal{K}^{n}_{\omega}(U)|\ge |D\KM (\mathrm{t}_*)|^{n-1} \, |U|,
    \end{equation}
    for all $n\ge 1$, where $|D\KM (\mathrm{t}_*)|>1$. Thus $\lim_{n\to \infty} |\mathcal{K}^{n}_{\omega}(U)|=\infty$, which contradicts the fact that the image of the Kasner map lies in the Kasner circle (of finite measure).

    To prove $(ii)$, we first claim that there exists a symbol $\omega \in \Sigma$ and $n\in \mb{N}_0$ such that $\mc{K}_\omega^n(U)$ is open, connected and $\mc{K}_\omega^n(U)\cap \partial \mathbf{A}$ contains at least 2 points.
    Suppose, towards a contradiction, that for every $\omega \in \Sigma$ and $n\in \mb{N}_0$, the set $\mc{K}_\omega^n(U)\cap \partial \mathbf{A}$ contain strictly less than 2 points. Note that in this case, there are symbols $\omega\in \Sigma$ such that the set $\KM_\omega^n(U)$ is open and connected for any $n\in \mb{N}_0$. 
    This implies that for any $n\in \mb{N}_0$, there is $\alpha(n)\in \{1,2,3\}$ such that $\mc{K}_\omega^n(U)\subset A_{\alpha(n)}$ and  $\mc{K}_\omega^n(U)\cap \partial \mathbf{A}$ contain strictly less than 2 points.


    
    Let $(N_k)_{k\ge 1}$ be the increasing sequence of all indices such that $\mc{K}_\omega^{N_k}(U)\cap \mathbf{A}\neq \emptyset$ which was proved in item $(i)$. 
    If $\mc{K}_\omega^{N_k}(U)\cap \partial \mathbf{A}=\emptyset$ for all $k$, then $\mc{K}_\omega^{N_k}(U)\subset \mathrm{int}(\mathbf{A})$ and thereby $\mc{K}_\omega^{N_k+1}(U)\subset A_{\alpha(n)}\setminus \mathbf{A}$. This yields an infinite sequence of maps with uniform expansion, which yields a contradiction akin to the proof of item $(i)$.
    Therefore $\mc{K}_\omega^{N_{*}}(U)\cap \partial \mathbf{A}$ contains exactly 1 point for some $N_*$.
    Thus, there is a symbol $\omega \in \Sigma$ such that  $\mc{K}_\omega^{N_{*}+1}(U)\cap \partial \mathbf{A}=\emptyset$ and we can repeat the above argument in order to construct an increasing sequence $(\tilde{N}_k)_{k\ge 1}$ of indices such that $\mc{K}_\omega^{\tilde{N}_k}(U)\cap \partial \mathbf{A}$ contains exactly 1 point. We have two cases:
    
    \begin{enumerate}
        \item[1.] If $\tilde{N}_{k+1}\neq \tilde{N}_k+1$ for all $k\in\mathbb{N}_0$, then $\mc{K}_\omega^{\tilde{N}_{k}+1}\subset A_{\alpha(\tilde{N}_{k}+1)}$ for all $k\in\mathbb{N}_0$. Hence, there is an infinite subsequence of indices for which the Kasner maps has uniform expansion, yielding a contradiction.
        {
        \item[2.] If $\tilde{N}_{k+1}=\tilde{N}_{k}+1$ for some  $k\in\mathbb{N}_0$, then
        \begin{equation}
            \mc{K}_\omega^{\tilde{N}_{k}+1}(U) \supseteq A_{\alpha(\tilde{N}_{k}+1)}\setminus \left(\mathbf{A}\cup \mc{K}_{\alpha(\tilde{N}_k)}(\mathbf{A}_{{\alpha(\tilde{N}_{k})} {\alpha(\tilde{N}_{k}+1)}})\right).
        \end{equation}
        for some $\omega\in\Sigma$, where $\KM_{\alpha(\tilde{N}_k)}$ is the map defined in the arc $A_{\alpha(\tilde{N}_k)}$ as in \eqref{KM_angle}. 
        Thus, 
        \begin{subequations}
            \begin{align}
            \mc{K}_\omega^{\tilde{N}_{k}+2}(U)&\supseteq  \mc{K}_{\alpha(\tilde{N}_k+1)}\left(A_{\alpha(\tilde{N}_{k}+1)}\setminus (\mathbf{A}\cup \mc{K}_{\alpha(\tilde{N}_k)}(\mathbf{A}_{{\alpha(\tilde{N}_{k})} {\alpha(\tilde{N}_{k}+1)}}))\right) \\
            &\supseteq A_{\alpha(\tilde{N}_{k})} \cap A_{\alpha(\tilde{N}_{k}+2)}.
            \end{align}
        \end{subequations}
        Therefore $\mc{K}_\omega^{\tilde{N}_{k}+2}(U)\cap \partial \mathbf{A}$ contains at least 2 points, which yields a contradiction.
        }
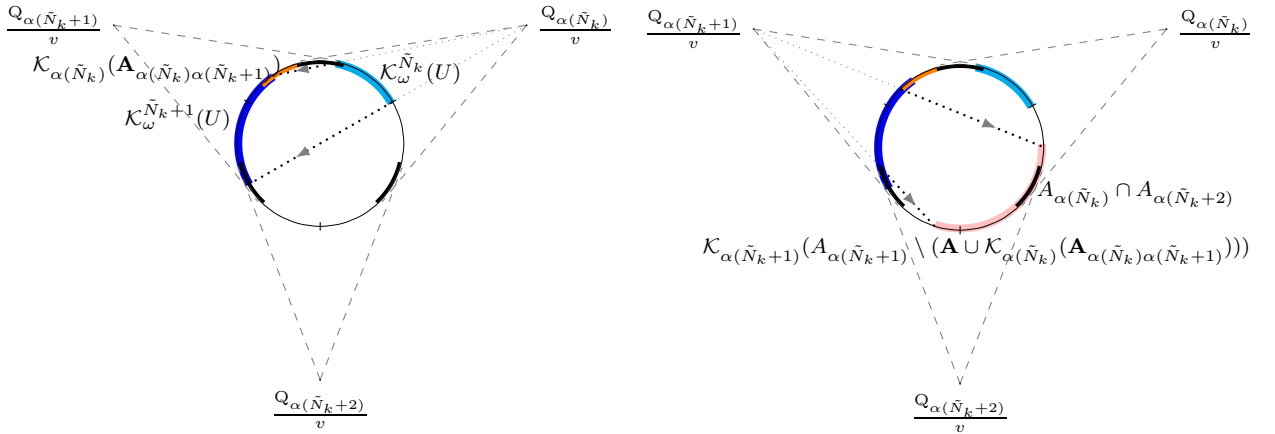
\begin{figure}[H]
\minipage[b]{0.49\textwidth}\centering
\begin{subfigure}\centering
    \begin{tikzpicture}[scale=1.1]

    \draw[rotate=120,dotted,color=gray] (0.995,0.09) -- (0,-2.85); 
    \draw[rotate=120,color=gray,dotted] (0,1) -- (0,-2.85);

    \draw[rotate=120,white,ultra thick] (0.995,0.09) -- (0.73,-0.7);
    \draw[rotate=120,white,ultra thick] (0,-1) -- (0,1);
    
    \draw [color=blue,line width=3pt, domain=3.75:5.24,variable=\t,smooth] plot ({0.98*sin(\t r)},{0.98*(-cos(\t r))});
    
    \draw [color=orange,line width=1.8pt, domain=-0.29:-0.78,variable=\t,smooth] plot ({0.98*sin(\t r)},{0.98*cos(\t r)});
    \node at (-1.95,0.9){\scriptsize{$\KM_{\alpha(\tilde{N}_k)}(\mathbf{A}_{\alpha(\tilde{N}_k)\alpha(\tilde{N}_k+1)})$}};
    
    \draw [color=cyan,line width=3pt, domain=2.09:2.95,variable=\t,smooth] plot ({0.98*sin(\t r)},{0.98*(-cos(\t r))});
    
    \draw[rotate=120,dotted, thick, postaction={decorate}] (0,-1) -- (0,1);
    \draw[rotate=120,dotted, thick, postaction={decorate}] (0.73,-0.7) -- (0.995,0.09);

    \node at (1.2,0.9) {\scriptsize{$\KM_{\omega}^{\tilde{N}_k}(U)$}};
    \node at (-1.7,0.34) {\scriptsize{$\KM_{\omega}^{\tilde{N}_k+1}(U)$}};
    
    \draw [line width=0.1pt,domain=0:6.28,variable=\t,smooth] plot ({sin(\t r)},{cos(\t r)});

    \draw [very thick, domain=-0.29:0.29,variable=\t,smooth] plot ({0.975*sin(\t r)},{0.975*cos(\t r)});
    \draw [rotate=120,very thick, domain=-0.29:0.29,variable=\t,smooth] plot ({0.975*sin(\t r)},{0.975*cos(\t r)});
    \draw [rotate=-120,very thick, domain=-0.29:0.29,variable=\t,smooth] plot ({0.975*sin(\t r)},{0.975*cos(\t r)});

    \draw[color=gray,dashed] (0,-2.85) -- (0.98,-0.29);
    \draw[color=gray,dashed] (0,-2.85) -- (-0.98,-0.29);

    \draw[color=gray,rotate=120,dashed] (0,-2.85) -- (0.98,-0.29);
    \draw[color=gray,rotate=120,dashed] (0,-2.85) -- (-0.98,-0.29);

    \draw[color=gray,rotate=240,dashed] (0,-2.85) -- (0.98,-0.29);
    \draw[color=gray,rotate=240,dashed] (0,-2.85) -- (-0.98,-0.29);

    \filldraw[gray] (0,-2.85) circle (0.1pt);






    \draw (0,-2.85) circle (0.1pt) node[anchor=north] {\scriptsize{$\frac{\mathrm{Q}_{\alpha(\tilde{N}_k+2)}}{v}$}};
    \draw[rotate=240] (0,-2.85) circle (0.1pt) node[anchor=east] {\scriptsize{$\frac{\mathrm{Q}_{\alpha(\tilde{N}_k+1)}}{v}$}};
    \draw[rotate=120]  (0,-2.85) circle (0.1pt) node[anchor=west] {\scriptsize{$\frac{\mathrm{Q}_{\alpha(\tilde{N}_k)}}{v}$}};

    \draw (0,-0.95) -- (0,-1.05);
    \draw[rotate=120] (0,-0.95) -- (0,-1.05);
    \draw[rotate=-120] (0,-0.95) -- (0,-1.05);
    
\end{tikzpicture}
    \addtocounter{subfigure}{-1}\captionof{subfigure}{\footnotesize{The (blue) set $\mc{K}_\omega^{\tilde{N}_{k}+1}(U)$ contains $A_{\alpha(\tilde{N}_{k}+1)}$ minus the (orange) set $\mc{K}_{\alpha(\tilde{N}_k)}(\mathbf{A}_{{\alpha(\tilde{N}_{k})} {\alpha(\tilde{N}_{k}+1)}})$ and the (bold) sets $\mathbf{A}$.}}
\end{subfigure}
\endminipage\hfill
\minipage[b]{0.49\textwidth}\centering
\begin{subfigure}\centering
    \begin{tikzpicture}[scale=1.1]


    \draw[rotate=-120,color=gray,dotted] (-0.5,0.86) -- (0,-2.85); 
    \draw[rotate=-120,color=gray,dotted] (0.98,0.22) -- (0,-2.85); 
    
    \draw[rotate=-120,white,ultra thick] (-0.25,-0.96) -- (-0.5,0.86);
    \draw[rotate=-120,white,ultra thick] (0.67,-0.75) -- (0.98,0.22);

    \draw [color=blue,line width=3pt, domain=3.75:5.24,variable=\t,smooth] plot ({0.98*sin(\t r)},{0.98*(-cos(\t r))});
    
    \draw [color=orange,line width=1.8pt, domain=-0.29:-0.78,variable=\t,smooth] plot ({0.98*sin(\t r)},{0.98*cos(\t r)});
    
    \draw [color=cyan,line width=3pt, domain=2.09:2.95,variable=\t,smooth] plot ({0.98*sin(\t r)},{0.98*(-cos(\t r))});

    \draw [color=pink,line width=3pt, domain=1.53:3.45,variable=\t,smooth] plot ({0.98*sin(\t r)},{0.98*cos(\t r)});
    \node at (0.2,-1.25){\scriptsize{$\mc{K}_{\alpha(\tilde{N}_k+1)}(A_{\alpha(\tilde{N}_{k}+1)}\setminus (\mathbf{A}\cup \mc{K}_{\alpha(\tilde{N}_k)}(\mathbf{A}_{{\alpha(\tilde{N}_{k})} {\alpha(\tilde{N}_{k}+1)}}))) $}};
    
    \draw[rotate=-120,dotted, thick, postaction={decorate}] (-0.25,-0.96) -- (-0.5,0.86);
    \draw[rotate=-120,dotted, thick, postaction={decorate}] (0.67,-0.75) -- (0.98,0.22);


    \draw [line width=0.1pt,domain=0:6.28,variable=\t,smooth] plot ({sin(\t r)},{cos(\t r)});

    \draw [very thick, domain=-0.29:0.29,variable=\t,smooth] plot ({0.975*sin(\t r)},{0.975*cos(\t r)});
    \draw [rotate=120,very thick, domain=-0.29:0.29,variable=\t,smooth] plot ({0.975*sin(\t r)},{0.975*cos(\t r)});
    \draw [rotate=-120,very thick, domain=-0.29:0.29,variable=\t,smooth] plot ({0.975*sin(\t r)},{0.975*cos(\t r)});

    \draw[color=gray,dashed] (0,-2.85) -- (0.98,-0.29);
    \draw[color=gray,dashed] (0,-2.85) -- (-0.98,-0.29);

    \draw[color=gray,rotate=120,dashed] (0,-2.85) -- (0.98,-0.29);
    \draw[color=gray,rotate=120,dashed] (0,-2.85) -- (-0.98,-0.29);

    \draw[color=gray,rotate=240,dashed] (0,-2.85) -- (0.98,-0.29);
    \draw[color=gray,rotate=240,dashed] (0,-2.85) -- (-0.98,-0.29);

    \filldraw[gray] (0,-2.85) circle (0.1pt);






    \draw (0,-2.85) circle (0.1pt) node[anchor=north] {\scriptsize{$\frac{\mathrm{Q}_{\alpha(\tilde{N}_k+2)}}{v}$}};
    \draw[rotate=240] (0,-2.85) circle (0.1pt) node[anchor=east] {\scriptsize{$\frac{\mathrm{Q}_{\alpha(\tilde{N}_k+1)}}{v}$}};
    \draw[rotate=120]  (0,-2.85) circle (0.1pt) node[anchor=west] {\scriptsize{$\frac{\mathrm{Q}_{\alpha(\tilde{N}_k)}}{v}$}};

    \draw (0,-0.95) -- (0,-1.05);
    \draw[rotate=120] (0,-0.95) -- (0,-1.05);
    \draw[rotate=-120] (0,-0.95) -- (0,-1.05);

    \node at (2.1,-0.57) {\scriptsize{$A_{\alpha(\tilde{N}_{k})} \cap A_{\alpha(\tilde{N}_{k}+2)}$}};
    
\end{tikzpicture}
    \addtocounter{subfigure}{-1}\captionof{subfigure}{\footnotesize{The set $\mc{K}_\omega^{\tilde{N}_{k}+2}(U)$ contains the image of the (blue) set without the (orange and bold) sets, and in turn contain $A_{\alpha(\tilde{N}_{k})} \cap A_{\alpha(\tilde{N}_{k}+2)}$. }}
\end{subfigure}
\endminipage
\captionof{figure}{In case $\tilde{N}_{k+1}=\tilde{N}_{k}+1$, the next iterate $\tilde{N}_{k}+2$ contains an overlap region.
}
\end{figure}
    \end{enumerate}

    We can now finish the proof of item $(ii)$, since we have proved that there exists a symbol $\omega \in \Sigma$ and $M\in \mb{N}_0$ such that $\mc{K}_\omega^M(U)$ is open, connected and $\mc{K}_\omega^M(U)\cap \partial \mathbf{A}$ contains at least 2 points. There are two cases. First, if $\KM_\omega^M(U) \supseteq A_\alpha \cap A_\beta$ for some $\alpha \neq \beta$, then the proof is complete. Second, if $\KM_\omega^M(U) \not \supseteq A_\alpha \cap A_\beta$ for all $\alpha \neq \beta$, thus $\KM_\omega^M(U)\supseteq A_\alpha \setminus \mathbf{A}$ for some $\alpha$. In this case, there exists a symbol $\omega\in \Sigma$ such that $\KM_\omega^{M+1}(U) \supseteq \KM_{\alpha}(A_\alpha \setminus \mathbf{A}) \supseteq A_\beta \cap A_\gamma $ where $\alpha,\beta,\gamma \in \{1,2,3\}$ are pairwise distinct and $\KM_{\alpha}$ is given similar to \eqref{KM_angle}.
    
    
    
%

\end{proof}

As a consequence of Lemma \ref{lem:intOVERLAP}, item (i), the set $\tilde{F}$ defined in \eqref{defofFtilde} is dense in the Kasner circle. Next, we prove the chaotic properties of the multi-valued Kasner map for $v\in (0,1/2)$.
    
\subsubsection*{Sensitivity to initial conditions}
    

    Without loss of generality, consider any point $p\in \KC$ and an open connected neighborhood $U\subset \KC$ of $p$. Lemma \ref{lem:intOVERLAP}, item (i), implies that for any $\omega\in \Sigma$, there exists a minimal $N\in \mathbb{N}$ such that $\KM_\omega^N(U)\cap \mathbf{A}\neq \emptyset$.
    Thus, there is a point $q\in U$ such that $\KM_\omega^N(q)\in A_\alpha\cap A_\beta$ and $\KM_\omega^N(p)\in A_\alpha\setminus A_\gamma$ for some $\alpha,\beta,\gamma\in \{1,2,3\},$ which are pairwise different.
    
    We then choose two different symbols $\omega,\omega^*\in \Sigma$ which only coincide at the first $N$ iterates, $\omega_n=\omega^*_n$ for $n=1,\ldots, N$, but differ at the next iterate such that the appropriate symbols within the triplets $\omega_{N+1}=(\mu_{N+1},\nu_{N+1},\zeta_{N+1})$ and $\omega^*_{N+1}=(\mu^*_{N+1},\nu^*_{N+1},\zeta^*_{N+1})$ corresponding to the intersection $A_\alpha\cap A_\beta$ attain different values: $\alpha$ or $\beta$. In other words, the $(N+1)$ iterate is given by
    \begin{equation}
        \KM_\omega^{N+1}(p)=\mc{K}_\alpha\left(\KM_\omega^N(p)\right) \quad \text{ and } \quad \KM_{\omega^*}^{N+1}(q)=\mc{K}_\beta\left(\KM_{\omega^*}^N(q)\right),
    \end{equation}
    where $\KM_\alpha$ and $\KM_\beta$ are the appropriate maps defined in the arcs $A_\alpha$ and $A_\beta$, as in \eqref{KM_angle}. See Figure \ref{fig:sens}. Therefore, $d\left(\KM^{N+1}_{\omega}(p),\KM^{N+1}_{\omega^*}(q)\right)>\delta$, where $\delta:=|\mathbf{A}|/3$.\footnote{Note that $\delta$ increases as the parameter $v\in (0,1/2)$ decreases. This indicates that as $v\in (0,1/2)$ decreases, the system has more sensitivity to initial data.}

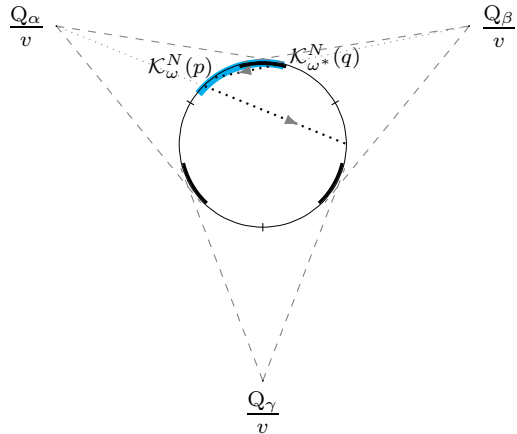
\begin{figure}[H]\centering
    \begin{tikzpicture}[scale=1.1]

    \draw[rotate=-120,color=gray,dotted] (-0.5,0.85) -- (0,-2.85);
    \draw[rotate=-120] (-0.5,0.85);
    \draw[rotate=-120] (-0.3,-0.95);
    
    \draw[rotate=120,color=gray,dotted] (0.985,0.1) -- (0,-2.85); 

    \draw[rotate=-120,white,ultra thick] (-0.5,0.85) -- (-0.26,-0.95);
    \draw[rotate=120,white,ultra thick] (0.985,0.1) -- (0.72,-0.7);
    
    \draw [color=cyan,line width=3pt, domain=-0.9:0.26,variable=\t,smooth] plot ({0.98*sin(\t r)},{0.98*cos(\t r)});

    \draw[rotate=-120,dotted, thick, postaction={decorate}] (-0.26,-0.95) -- (-0.5,0.85);
    \draw[rotate=120,dotted, thick, postaction={decorate}] (0.72,-0.7) -- (0.985,0.1);
    
    \node at (0.75,1.05) {\scriptsize{$\KM_{\omega^*}^{N}(q)$}};
    \node at (-0.98,0.94) {\scriptsize{$\KM_{\omega}^{N}(p)$}};
    



    \draw[color=gray,dashed] (0,-2.85) -- (0.98,-0.29);
    \draw[color=gray,dashed] (0,-2.85) -- (-0.98,-0.29);

    \draw[color=gray,rotate=120,dashed] (0,-2.85) -- (0.98,-0.29);
    \draw[color=gray,rotate=120,dashed] (0,-2.85) -- (-0.98,-0.29);

    \draw[color=gray,rotate=240,dashed] (0,-2.85) -- (0.98,-0.29);
    \draw[color=gray,rotate=240,dashed] (0,-2.85) -- (-0.98,-0.29);

    \filldraw[gray] (0,-2.85) circle (0.1pt);
    
    \draw [line width=0.1pt,domain=0:6.28,variable=\t,smooth] plot ({sin(\t r)},{cos(\t r)});

    \draw [very thick, domain=-0.29:0.29,variable=\t,smooth] plot ({0.975*sin(\t r)},{0.975*cos(\t r)});
    \draw [rotate=120,very thick, domain=-0.29:0.29,variable=\t,smooth] plot ({0.975*sin(\t r)},{0.975*cos(\t r)});
    \draw [rotate=-120,very thick, domain=-0.29:0.29,variable=\t,smooth] plot ({0.975*sin(\t r)},{0.975*cos(\t r)});

    \draw (0,-2.85) circle (0.1pt) node[anchor=north] {$\frac{\mathrm{Q}_\gamma}{v}$};
    \draw[rotate=240] (0,-2.85) circle (0.1pt) node[anchor=east] {$\frac{\mathrm{Q}_\alpha}{v}$};
    \draw[rotate=120]  (0,-2.85) circle (0.1pt) node[anchor=west] {$\frac{\mathrm{Q}_\beta}{v}$};

    \draw (0,-0.95) -- (0,-1.05);
    \draw[rotate=120] (0,-0.95) -- (0,-1.05);
    \draw[rotate=-120] (0,-0.95) -- (0,-1.05);

\end{tikzpicture}
\captionof{figure}{After the (cyan) set $\KM^N_\omega(U)$ intersects the (bold) overlap region $A_\alpha \cap A_\beta$, we choose different symbols $\omega,\omega^*$ such that the corresponding further images of $\KM_\omega^N(p)\in A_\alpha\setminus A_\gamma$ and $\KM_{\omega^*}^N(q)\in A_\alpha\cap A_\beta$ are $\delta:=\mathbf{A}/3$ apart. Note that $\KM_\omega^N(p)$ can be either in $A_\alpha \cap A_\beta$ or $A_\alpha\setminus \mathbf{A}$.
}\label{fig:sens}
\end{figure}
%



\subsubsection*{Topological Transitivity}

    Consider open sets $U,V\subseteq \KC$. Without loss of generality, $U$ is connected. 
    We claim that for any $\alpha,\beta \in \{1,2,3\},\ \alpha \neq \beta$, there are $\omega_* \in \Sigma$ and $N_*\in \mb{N}_0$ such that $\KM_{\omega_*}^{N_*}(A_\alpha \cap A_\beta)\cap V\neq \emptyset$.
    Topological transitivity follows directly from such claim, since due to Lemma \ref{lem:intOVERLAP}, item $(ii)$, there are $\omega \in \Sigma,\ M\in \mb{N}_0$ such that $\KM^M_\omega(U) \supseteq A_\alpha \cap A_\beta$ for some $\alpha, \beta \in \{1,2,3\}, \alpha\neq \beta$, and thus
    \begin{equation}
        \mc{K}_{\omega_*}^{N_*}\circ \mc{K}_{\omega}^M(U)\cap V\neq \emptyset.
    \end{equation}
%
%
    Let us now prove the claim.
    For the sake of notation, we suppose $\alpha=2,\beta=3$, as the other cases are analogous. Consider the symbols $\omega'=(\mu'_k,\nu'_k,\zeta'_k)_{k\in\mathbb{N}_0},\ \omega''=(\mu''_k,\nu''_k,\zeta''_k)_{k\in\mathbb{N}_0}\in \Sigma$ which are periodic in the the second coordinate according to
    \begin{equation}\label{symbols}
    \nu'_n=\begin{cases}2\text{, if }n\text{ is odd}\\
    3\text{, if }n\text{ is even}\end{cases} \quad \text{ and } \quad
    \nu''_n=\begin{cases}3\text{, if }n\text{ is odd}\\
    2\text{, if }n\text{ is even.}\end{cases}
    \end{equation}
    There exists an $m\in \mathbb{N}_0$ such that the union of iterates $\mc{K}^k_{\omega'}(A_2 \cap A_3)$ and $\mc{K}^k_{\omega''}(A_2 \cap A_3)$ for $k=1,\ldots,m$ covers the arcs $A_2$ and $A_3$. Moreover, the next iterate of these sets will cover $A_1\setminus \mathbf{A}$. Consequently the union of all iterates $k=1,\ldots,m+1$ of the set $A_2 \cap A_3$ under the maps $\mc{K}_{\omega'}$ and $\mc{K}_{\omega''}$ will cover the Kasner circle. See Figure \ref{fig:trans}. Consequently, there is an $N_*\in\mathbb{N}_0$ such that $N_*\leq m+1$ and $\mc{K}_{\omega'}^{N_*}(A_2 \cap A_3)\cap V\neq \emptyset$ for any set $V\subseteq \KC$.
    

%
\begin{figure}[H]
\minipage[b]{0.5\textwidth}\centering
\begin{subfigure}\centering
    \begin{tikzpicture}[scale=1.1]

    \draw[rotate=120,color=gray,dotted] (0.975,0.25) -- (0,-2.85); 
    \node at (1.6,0.85) {\scriptsize{$\KM_{\omega'}(A_2\cap A_3)$}};
    \draw[rotate=-120,color=gray,dotted] (-0.975,0.25) -- (0,-2.85); 
    \node at (-1.6,0.85) {\scriptsize{$\KM_{\omega''}(A_2\cap A_3)$}};
    
    \draw[rotate=120,white,ultra thick] (0.975,0.25) -- (0.67,-0.7);
    \draw[rotate=-120,white,ultra thick] (-0.975,0.25) -- (-0.67,-0.7);
        
    \draw [color=cyan,line width=1.8pt, domain=0.29:0.78,variable=\t,smooth] plot ({0.98*sin(\t r)},{0.98*cos(\t r)});
    \draw [color=cyan,line width=1.8pt, domain=-0.29:-0.78,variable=\t,smooth] plot ({0.98*sin(\t r)},{0.98*cos(\t r)});

    \draw[rotate=120,dotted, thick, postaction={decorate}] (0.67,-0.7) -- (0.975,0.25);
    \draw[rotate=-120,dotted, thick, postaction={decorate}] (-0.67,-0.7) -- (-0.975,0.25);

    \draw [line width=0.1pt,domain=0:6.28,variable=\t,smooth] plot ({sin(\t r)},{cos(\t r)});

    \draw [very thick, domain=-0.29:0.29,variable=\t,smooth] plot ({0.975*sin(\t r)},{0.975*cos(\t r)});
    \draw [rotate=120,very thick, domain=-0.29:0.29,variable=\t,smooth] plot ({0.975*sin(\t r)},{0.975*cos(\t r)});
    \draw [rotate=-120,very thick, domain=-0.29:0.29,variable=\t,smooth] plot ({0.975*sin(\t r)},{0.975*cos(\t r)});

    \draw[color=gray,dashed] (0,-2.85) -- (0.98,-0.29);
    \draw[color=gray,dashed] (0,-2.85) -- (-0.98,-0.29);

    \draw[color=gray,rotate=120,dashed] (0,-2.85) -- (0.98,-0.29);
    \draw[color=gray,rotate=120,dashed] (0,-2.85) -- (-0.98,-0.29);

    \draw[color=gray,rotate=240,dashed] (0,-2.85) -- (0.98,-0.29);
    \draw[color=gray,rotate=240,dashed] (0,-2.85) -- (-0.98,-0.29);

    \filldraw[gray] (0,-2.85) circle (0.1pt);






    \draw (0,-2.85) circle (0.1pt) node[anchor=north] {$\frac{\mathrm{Q}_1}{v}$};
    \draw[rotate=240] (0,-2.85) circle (0.1pt) node[anchor=east] {$\frac{\mathrm{Q}_2}{v}$};
    \draw[rotate=120]  (0,-2.85) circle (0.1pt) node[anchor=west] {$\frac{\mathrm{Q}_3}{v}$};

    \draw (0,-0.95) -- (0,-1.05);
    \draw[rotate=120] (0,-0.95) -- (0,-1.05);
    \draw[rotate=-120] (0,-0.95) -- (0,-1.05);

    \node at (0,1.2) {\scriptsize{$A_2\cap A_3$}};
\end{tikzpicture}
    \addtocounter{subfigure}{-1}\captionof{subfigure}{\footnotesize{First iterate using the words $\omega'$ and $\omega''$.}}
\end{subfigure}
\endminipage\hfill
\minipage[b]{0.5\textwidth}\centering
\begin{subfigure}\centering
    \begin{tikzpicture}[scale=1.1]

    
    \draw[rotate=120,color=gray,dotted] (0.975,0.25) -- (0,-2.85); 
    \draw[rotate=-120,color=gray,dotted] (-0.975,0.25) -- (0,-2.85); 
    
    \draw[rotate=120,white,ultra thick] (0.975,0.25) -- (0.67,-0.7);
    \draw[rotate=-120,white,ultra thick] (-0.975,0.25) -- (-0.67,-0.7);

    \draw[rotate=120,color=gray,dotted] (0.51,0.87) -- (0,-2.85); 
    \draw[rotate=-120,color=gray,dotted] (-0.51,0.87) -- (0,-2.85); 
    
    \draw[rotate=120,white,ultra thick] (0.26,-0.97) -- (0.51,0.87);
    \draw[rotate=-120,white,ultra thick] (-0.26,-0.97) -- (-0.51,0.87);
                
    \draw [color=cyan,line width=1.8pt, domain=0.29:0.78,variable=\t,smooth] plot ({0.98*sin(\t r)},{0.98*cos(\t r)});
    \draw [color=cyan,line width=1.8pt, domain=-0.29:-0.78,variable=\t,smooth] plot ({0.98*sin(\t r)},{0.98*cos(\t r)});

    \draw [color=blue,line width=1.8pt, domain=-0.78:-1.56,variable=\t,smooth] plot ({0.98*sin(\t r)},{0.98*cos(\t r)});\node at (-1.8,0.5) {\scriptsize{$\KM_{\omega'}^2(A_2\cap A_3)$}};
    \draw [color=blue,line width=1.8pt, domain=0.78:1.56,variable=\t,smooth] plot ({0.98*sin(\t r)},{0.98*cos(\t r)});\node at (1.8,0.5) {\scriptsize{$\KM_{\omega''}^2(A_2\cap A_3)$}};

    \draw[rotate=120,dotted, thick, postaction={decorate}] (0.67,-0.7) -- (0.975,0.25);
    \draw[rotate=-120,dotted, thick, postaction={decorate}] (-0.67,-0.7) -- (-0.975,0.25);
    
    \draw[rotate=120,dotted, thick, postaction={decorate}] (0.26,-0.97) -- (0.51,0.87);
    \draw[rotate=-120,dotted, thick, postaction={decorate}] (-0.26,-0.97) -- (-0.51,0.87);


    \draw [line width=0.1pt,domain=0:6.28,variable=\t,smooth] plot ({sin(\t r)},{cos(\t r)});

    \draw [very thick, domain=-0.29:0.29,variable=\t,smooth] plot ({0.975*sin(\t r)},{0.975*cos(\t r)});
    \draw [rotate=120,very thick, domain=-0.29:0.29,variable=\t,smooth] plot ({0.975*sin(\t r)},{0.975*cos(\t r)});
    \draw [rotate=-120,very thick, domain=-0.29:0.29,variable=\t,smooth] plot ({0.975*sin(\t r)},{0.975*cos(\t r)});

    \draw[color=gray,dashed] (0,-2.85) -- (0.98,-0.29);
    \draw[color=gray,dashed] (0,-2.85) -- (-0.98,-0.29);

    \draw[color=gray,rotate=120,dashed] (0,-2.85) -- (0.98,-0.29);
    \draw[color=gray,rotate=120,dashed] (0,-2.85) -- (-0.98,-0.29);

    \draw[color=gray,rotate=240,dashed] (0,-2.85) -- (0.98,-0.29);
    \draw[color=gray,rotate=240,dashed] (0,-2.85) -- (-0.98,-0.29);

    \filldraw[gray] (0,-2.85) circle (0.1pt);






    \draw (0,-2.85) circle (0.1pt) node[anchor=north] {$\frac{\mathrm{Q}_1}{v}$};
    \draw[rotate=240] (0,-2.85) circle (0.1pt) node[anchor=east] {$\frac{\mathrm{Q}_2}{v}$};
    \draw[rotate=120]  (0,-2.85) circle (0.1pt) node[anchor=west] {$\frac{\mathrm{Q}_3}{v}$};

    \draw (0,-0.95) -- (0,-1.05);
    \draw[rotate=120] (0,-0.95) -- (0,-1.05);
    \draw[rotate=-120] (0,-0.95) -- (0,-1.05);

    \node at (0,1.2) {\scriptsize{$A_2\cap A_3$}};
\end{tikzpicture}
    \addtocounter{subfigure}{-1}\captionof{subfigure}{\footnotesize{Second iterate using the words $\omega'$ and $\omega''$.}}
\end{subfigure}
\endminipage
\captionof{figure}{
The union of iterates of $A_1\cap A_2$ by the maps $\mc{K}_{\omega'}$ and $\mc{K}_{\omega''}$ will eventually cover the circle.
}\label{fig:trans}
\end{figure}
\subsubsection*{Density of Periodic orbits}


    Given any point $p\in \KC$ and an open neighborhood $U\subseteq \KC$ of $p$, we will show that for sufficiently small $\varepsilon>0$ there exist a closed neighborhood ${N}_\varepsilon \subseteq U$ 
    and a contraction in ${N}_\varepsilon$, which is constructed by 
    pre-images of maps $\mc{K}_{\mu \nu \zeta}$ in \eqref{Kmnzdef}. Hence Banach's fixed point theorem implies that there is a periodic point in ${N}_\varepsilon$ and thereby also in $U$.

    Recall the construction in the proof of topological transitivity, which guarantees that there are symbols $\omega',\omega''\in\Sigma$ such that the union of iterates $\mc{K}^k_{\omega'}(A_2 \cap A_3)$ and $\mc{K}^k_{\omega''}(A_2 \cap A_3)$ for $k=1,\ldots,m+1$ covers the Kasner circle, see Figure \ref{fig:trans}. We can repeat this procedure for the other overlap regions. Indeed, consider the symbols $\tilde{\omega},\tilde{\tilde{\omega}}\in \Sigma$ which are periodic in 
    $\tilde{\mu}_n,\tilde{\tilde{\mu}}_n$ with alternating numbers $1,2$, similar to \eqref{symbols}. Therefore, the union of iterates $\mc{K}^k_{\tilde{\omega}}(A_1 \cap A_2)$ and $\mc{K}^k_{\tilde{\tilde{\omega}}}(A_1 \cap A_2)$ for $k=1,\ldots,m+1$ covers the Kasner circle.
    Similarly, consider the symbols $\dot{\omega},\ddot{\omega}\in \Sigma$ which are periodic in 
    $\dot{\zeta}_n,\ddot{\zeta}_n$ with alternating numbers $1,3$, and thus the union of $\mc{K}^k_{\dot{\omega}}(A_1 \cap A_3)$ and $\mc{K}^k_{\ddot{\omega}}(A_1 \cap A_3)$ for $k=1,\ldots,m+1$ covers the circle.
    
    For sufficiently small $\epsilon>0$, we choose a closed neighborhood $N_\epsilon\subseteq U$ which is strictly contained in only one set of the following family:
    \begin{align}\label{family}
    \left\{ \mc{K}^k_{{\omega_1}}(A_\alpha \cap A_\beta)\cap \mc{K}^n_{{\omega_2}}(A_\alpha \cap A_\gamma) \,\, \Big|  
    \begin{array}{cc}
        (\alpha\beta\gamma) \text{ is a permutation of } (123)\\
        \omega_1,\omega_2\in \{\omega',\omega'',\tilde{\omega},\tilde{\tilde{\omega}},\dot{\omega},\ddot{\omega}\}  \\
        k,n\in \{1,\ldots,m+1\}
    \end{array} \right\}.
    \end{align}
    
    Hence, for any distinct $\alpha,\beta \in \{1,2,3\}$, there are $N^*\leq m+1$ and $\omega_*\in \{\omega',\omega'',\tilde{\omega},\tilde{\tilde{\omega}},\dot{\omega},\ddot{\omega}\}$ such that $N_\epsilon\subseteq \mc{K}^{N^*}_{{\omega_*}}(A_\alpha \cap A_\beta)$.
    Moreover, due to Lemma \ref{lem:intOVERLAP}, item $(ii)$, there are $\omega \in \Sigma, M\in \mb{N}_0$ such that $\KM^M_\omega(N_\epsilon) \supseteq A_\alpha \cap A_\beta$ for some distinct $\alpha, \beta \in \{1,2,3\}$. Therefore, choosing the appropriate pre-images yields the following contraction,
    \begin{equation} \label{contract}
        \KM^{-M}_\omega \circ \KM^{-N^*}_{\omega_*}: N_\epsilon \to N_\epsilon,
    \end{equation}
    and thereby we obtain a fixed point. 
    Note that less than $N^*+M$ maps in the concatenation \eqref{contract} are non-uniform contractions, which implies that the whole concatenation is a contraction.
    Therefore, there is a periodic orbit with period $N_*+M$ and symbol which repeats periodically $(w^*_{1}\ldots w^*_{N^*} w_1\ldots w_M)$.

{
\color{black}

\section{Conclusion}\label{sec:disc}

We give an overview on our proof methodology, some consequences, alternative approaches and possible future directions. We also compare our results with the existing literature.

The proof in Section \ref{sec:pf} is based on the technical Lemma \ref{lem:intOVERLAP}, which shows that iterated images of an arbitrary open set of the Kasner circle eventually cover some connected component of the multi-valued set $\mathbf{A}$ that possess positive length for $v\in (0,1/2)$. 
This allows us to prove Theorem \ref{mainthm} using a constructive approach, due to the freedom to choose among different iterates once one component of the multi-valued set $\mathbf{A}$ is covered.
This readily yields sensitivity to initial conditions, since there are two possible immediate iterates after covering a component of $\mathbf{A}$, which are at least $|\mathbf{A}|/3$ apart.
To prove topological transitivity, we again construct two different iterates after a component of $\mathbf{A}$ is covered and compare their eventual images: their union cover the whole circle, and thus they also cover any particular set. 
The density of periodic orbits is obtained by considering appropriate pre-images of an arbitrary set until it is eventually contained in $\mathbf{A}$, and then concatenating with appropriate pre-images that eventually return to the initial arbitrary set. Such a concatenation yields a uniform contraction and thereby a fixed point by Banach's fixed point theorem. 

Note our proof relies on the specific definition of chaos realization for multi-valued maps in the introduction, which considers the three conditions of chaos for \emph{some} $\omega\in\Sigma$. 
One may wish to investigate a more rigid notion of chaos realization for the multi-valued map, such as requiring the three conditions to be satisfied for \emph{every} element in the space $\Sigma$, or \emph{almost every} symbol in $\Sigma$ with respect to the product measure of uniform distributions of each trajectory to be chosen at time $n$. However, this would probably require a non-constructive approach. 

Notice the multi-valued set $\mathbf{A}$ degenerates at the Taub points towards GR, i.e., $|\mathbf{A}|\to 0$ as $v\to 1/2$.
In particular, the technical Lemma \ref{lem:intOVERLAP} fails and thus our proof for the subcritical case, $v\in (0,1/2)$, does not provide a satisfactory approximation of generic dynamics for GR. This is in contrast with the case supercritical case, $v\in (1/2,1)$, that yields a limiting symbolic dynamics description for the generic dynamics of GR; see \cite[Appendix C]{HellLappicyUggla}.
In this regard, we believe that the subcritical regime may provide a suitable approximation of the non-generic dynamics in GR by unfolding the degenerate Taub points with the parameter $v$, as known rigorous results in GR fail nearby the Taub points, for example \cite{Liebscher,Beguin}. 
Similarly, the proof fails in the limit $v\to 0$, since the Kasner map becomes an isometry and there is no expansion, see \cite{LappicyLessard}.

Even though the overall dynamics of the Kasner circle map ${\cal K}$ is far from being understood, there are special features which can be compared with the supercritical case, $v\in (1/2,1)$. Consider the set $\tilde{C}$ of points in the Kasner circle $\mathrm{K}^{\ocircle}$
for which all iterates of ${\cal K}$ consist of
exactly one positive eigenvalue in the $N_\alpha$ variables
, see Figure
~\ref{FIG:KASNERMAPS}, i.e.,
\begin{equation}\label{defofCtilde}
\tilde{C}:= \{ p\in \mathrm{K}^{\ocircle} \text{ $ | $ } \mathcal{K}^n(p)
\notin \mathrm{int}((A_1\cap A_2)\cup (A_1\cap A_3)\cup (A_2\cap A_3))  \text{ for all }  n \in \mathbb{N}_0 \}.
\end{equation}
This set is given by the points that never reach the overlaps
$\mathrm{int}(A_\alpha \cap A_\beta)$. The map $\mathcal{K}$ is thereby not multi-valued on the set $\tilde{C}$, and thus $\mathcal{K}$ is a well-defined map.
Furthermore, the set $\tilde{C}$ is not empty since, e.g., there are two
(physically equivalent) period 3 cycles, see \cite{HellLappicyUggla}.

The complement of the set $\tilde{C}$ in $\mathrm{K}^{\ocircle}$, which thereby consists of points that eventually are in the multi-valued regime, is given by
\begin{equation}\label{defofFtilde}
\tilde{F}:= \{ p\in \mathrm{K}^{\ocircle} \text{ $ | $ }
\mathcal{K}^n(p) \in  \mathrm{int}((A_1\cap A_2)\cup (A_1\cap A_3)\cup (A_2\cap A_3)) \text{ for some } n \in \mathbb{N}_0\},
\end{equation}
%
%
Splitting the dynamics in $\mathrm{K}^\ocircle$ into two
disjoint invariant sets, $\tilde{C}$ and $\tilde{F}$, is
a first step to understand the overall dynamics of the multi-valued map $\mathcal{K}$.
In particular, as a consequence of the Lemma \ref{lem:intOVERLAP}, the set $\tilde{F}$ is dense and thereby the generic dynamics occurs in such a set, akin to the generic dynamics outside the Cantor set in the supercritical case, see \cite{HellLappicyUggla}.

Instead of defining the iterates of the Kasner map as a skew-product dynamical system in \eqref{defKM}, which selects appropriate choices of maps by means of a symbol, we can consider the dynamics of all possible iterates. In the spirit of~\cite{Hutch81}, we may define the iterates of $\mathcal{K}$ by the Hutchinson operator,
\begin{equation}\label{KasnerHutchinson}
\mathcal{K}^n(p):=
\bigcup_{\substack{\mu_k=1,2 ; \hspace{0.1cm} \nu_k=2,3 ;
\hspace{0.1cm} \zeta_k=1,3\\ \text{for }  k=1,...,n}}
{\cal K}_{\mu_n\nu_n\zeta_n}\circ ...\circ {\cal K}_{\mu_1\nu_1\zeta_1}(p).
\end{equation}
In general, the Hutchinson operator for uniform contractions guarantees the existence of an attractor, i.e., 
a nonempty compact set $\mathbb{A}\subseteq \mathrm{K}^\ocircle$ such that $\lim_{n\to \infty} \mathcal{K}^n(C)=\mathbb{A}$, with respect to the Hausdorff metric for every nonempty compact set $C\subseteq \mathrm{K}^\ocircle$, see \cite{Hutch81}. 
In case the maps are continuous, but the contractions are not uniform or weakly hyperbolic, there is still an attractor $\mathbb{A}$ when the diameter of iterates of the phase-space converge to zero for some sequence of symbols, see \cite{MatDiaz,ArJuSa17}.
Note that in the usual physical time direction (i.e., the reverse of the present time direction), the Kasner map becomes a non-uniform contraction. However, the discontinuities of the Kasner maps still pose problems in guaranteeing the existence of an attractor and describing its internal dynamics. Thus, a finer and more robust dynamical description remains an open problem. We conjecture that a dense attractor (with chaotic dynamics) exists for the multi-valued Kasner map, which is corroborated by numerical experiments.
In particular, the relation of a conjectured attractor $\mathbb{A}$ with the dense invariant set $\tilde{F}$ is an open issue.

Lastly, we compare our result with the existing physical literature on chaotic aspects of spatially homogeneous Ho\v{r}ava-Lifshitz models. 

In \cite{GiKam17}, the authors consider a fixed kinetic parameter $\lambda=1$ with two types of dominant potential terms: a quadratic and a cubic one (resp. $k_2 R^2+k_3 R^i\!_jR^j\!_i$ and $k_6 R^3+\ldots$ in \eqref{calV}). 
For quadratic potentials, they concluded that the dynamical regime towards the singularity is chaotic.
For cubic potentials, they concluded that the dynamical regime towards the singularity is oscillating, but not chaotic. 
In our setting, these quadratic and cubic potentials with $\lambda=1$ correspond to the parameter values $v=1/8$ and $v=0$, respectively, according to \eqref{v's}.
Our results are in agreement with these conclusions.
Indeed, Theorem \ref{mainthm} is valid for $v=1/8$, whereas our proof fails for $v=0$. See \cite{LappicyLessard} for a more detailed discussion in the case that $v=0$, where chaos of the Kasner map is suppressed, but the Bianchi type IX has continuous periodic orbits (which are not heteroclinic chains).

In \cite{Bakas10}, the authors consider a perturbed kinetic parameter $\lambda\neq 1$ with quadratic dominant curvature terms which corresponds to $k_2 R^2 + k_3 R^i\!_jR^j\!_i + k_4 R^i\!_jC^j\!_i + k_5 C^i\!_jC^j\!_i $ in \eqref{calV}. 
For large anisotropy, they argue that the quadratic Cotton potential, $k_5 C^i\!_jC^j\!_i$, is dominant towards the singularity. 
There are similarities and differences from GR: solutions will keep bouncing from one Kasner state to another (with a different bounce law from GR), but the dynamics is not chaotic.
In our setting, such a dominant quadratic Cotton potential corresponds to the case $v=0$, due to \eqref{v's}, and thus their conclusions are in agreement with our results.

Hence our framework proposes a unified approach to analyze different dominant potential terms. Not only it comprises the existing findings in the literature for specific dominant potentials, but also describes the generic dynamics for a broad range of potentials which is given by a unified parameter $v\in (0,1)$.

\appendix

\section{Derivation of the ODE model}\label{app}
We deduce the evolution equations \eqref{full:subs} from the action \eqref{action} following \cite[Appendix A]{HellLappicyUggla}.
For the vacuum HL class~A Bianchi models, the action~\eqref{action}
expressed in terms of a symmetry adapted spatial (left-invariant) co-frame $\{{\omega}^1,{\omega}^2,{\omega}^3\}$  yields the field equations for the associated metric~\eqref{genmetric}.
Expressing the components of the spatial metric in such a symmetry adapted
spatial co-frame leads to that they become purely time-dependent in diagonal form, see~\cite{waiell97} and references therein. 
Setting the shift vector $N_i$ in~\eqref{genmetric} to zero, 
the diagonalized vacuum spatially homogeneous class~A metrics are given by
\begin{equation}\label{threemetric}
\mathbf{g} = -N^2(t)dt\otimes
dt + g_{11}(t)\:{\omega}^1\otimes {\omega}^1 +
g_{22}(t)\:{\omega}^2\otimes {\omega}^2 +
g_{33}(t)\:{\omega}^3\otimes {\omega}^3,
\end{equation}
where the lapse $N(t)$ is a non-zero function determining
the particular choice of time variable.

In order to obtain simple Hamiltonian equations, we first focus on the kinetic
part ${\cal T}$ in equation~\eqref{kin}, which can be written as
\begin{equation}\label{kin2}
{\cal T} = (K^1\!_1)^2 + (K^2\!_2)^2 + (K^3\!_3)^2 - \lambda (K^1\!_1 + K^2\!_2 + K^3\!_3)^2,
\end{equation}
where the extrinsic curvature is given by
$(K_{11}, K_{22}, K_{33}) = (\dot{g}_{11}, \dot{g}_{22}, \dot{g}_{33})/(2N)$ such that $\dot{}$ denotes a derivative with respect to $t$, and thus raising one of the indices, we obtain that $(K^1\!_1, K^2\!_2, K^3\!_3) 
= \left(\dot{g}_{11}/g_{11}, \dot{g}_{22}/g_{22}, \dot{g}_{33}/g_{33}\right)/(2N)$.
%

To simplify ${\cal T}$, we make a variable transformation
from the metric components to the variables $\beta^0,\beta^+,\beta^-$,
first introduced by Misner~\cite{mis69a,mis69b,grav73},
\begin{equation}\label{Misnerbeta}
g_{11} = e^{2( \beta^0 - 2\beta^+)}, \qquad
g_{22} = e^{2(\beta^0 + \beta^+ + \sqrt{3}\beta^-)},\qquad
g_{33} = e^{2(\beta^0 + \beta^+ - \sqrt{3}\beta^-)}.
\end{equation}
This results in that ${\cal T}$ in equation~\eqref{kin2} takes the form
\begin{equation}\label{calT}
{\cal T} = \dfrac{6}{N^2}
\left[-\left(\frac{3\lambda-1}{2}\right)(\dot{\beta}^0)^2 + (\dot{\beta}^+)^2 + (\dot{\beta}^-)^2\right].
\end{equation}
Note that the character of the quadratic form~\eqref{calT} changes
when $\lambda = 1/3$. Since we are interested
in continuously deforming the GR case $\lambda = 1$, we restrict considerations to
$\lambda > 1/3$. To simplify the kinetic part further, we introduce a new variable
$\beta^\lambda$ and a density-normalized lapse function $\mathcal{N}$, defined by
\begin{equation}\label{betaLAMBDAandNcal}
\beta^\lambda := \sqrt{\frac{3\lambda-1}{2}}\beta^0,\qquad
{\cal N} := \frac{N}{12\sqrt{g}},
\end{equation}
where $g = g_{11}g_{22}g_{33} = \exp(6\beta^0)$ is the determinant of the spatial metric
in the symmetry adapted co-frame, which leads to,
\begin{equation}
\sqrt{g}N{\cal T} = \dfrac{1}{2 {\cal N}}\left[-(\dot{\beta}^\lambda)^2 + (\dot{\beta}^+)^2 + (\dot{\beta}^-)^2\right].
\end{equation}

It is convenient to define $T := \frac{\sqrt{g} N}{\mathcal{N}}\mathcal{T} = 12g{\cal T}$, so that ${\cal N}T$ is the kinetic part of the Lagrangian for
the present spatially homogeneous models, in
analogy with the GR case, see e.g., ch. 10 in~\cite{waiell97}.
The density-normalized lapse ${\cal N}$ is kept in the kinetic term
${\cal N}T$, since it is needed in order to obtain the Hamiltonian
constraint, which is accomplished by varying ${\cal N}$ in the Hamiltonian.

To proceed to a Hamiltonian description, we introduce the canonical momenta
\begin{equation}
p_\lambda := -\frac{\dot{\beta}^\lambda}{{\cal N}}, \qquad p_\pm :=\frac{\dot{{\beta}}^\pm}{{\cal N}}.
\end{equation}
This leads to that $T$ takes the form
\begin{equation}\label{Tkinetics}
T = \frac12\left(- p_\lambda^2 + p_+^2 + p_-^2\right).
\end{equation}

Similarly to the treatment of the kinetic part, we define
\begin{equation}
V := \sqrt{g}N {\cal V}/{\cal N}=12 g {\cal V}.
\end{equation}
Due to~\eqref{calV},
\begin{equation}\label{VHL}
V = {}^1V + {}^2V + {}^3V + {}^4V + {}^5V + {}^6V + \dots ,
\end{equation}
where
\begin{subequations}\label{pots}
\begin{alignat}{3}
{}^1V &:= 12k_1 gR, &\qquad {}^2V &:= 12k_2 gR^2,  &\qquad
{}^3V &:= 12k_3 g R^i\!_jR^j\!_i,\\
{}^4V &:= 12k_4 g R^i\!_jC^j\!_i,
&\qquad {}^5V &:= 12k_5 g C^i\!_jC^j\!_i, &\qquad {}^6V &:= 12k_6 gR^3.
\end{alignat}
\end{subequations}
The superscripts on ${}^AV$ (where $A = 1,\dots,6$) thereby coincide with the
subscripts of the constants $k_A$ in~\eqref{calV}.

Based on~\eqref{action}, this leads to a Hamiltonian $H$ given by
\begin{equation}\label{LambdaRham}
H := \sqrt{g}N({\cal T} + {\cal V}) = {\cal N}(T + V) = 0,
\end{equation}
where $T$ only depends on the canonical momenta $p_\lambda$, $p_\pm$,
given by~\eqref{Tkinetics}, and $V$ only depends on
$\beta^\lambda$, $\beta^\pm$, given by~\eqref{VHL} and~\eqref{pots}.

In order to derive the ordinary differential equations for these models
via the Hamiltonian equations in terms of the variables
$\beta^\lambda$, $\beta^\pm$ and the canonical momenta $p_\lambda$, $p_\pm$,
we need to compute each ${}^AV(\beta^\lambda,\beta^\pm)$.
We proceed with the simplest case that minimally modifies vacuum GR in the present context, the vacuum $\lambda$-$R$ models~\cite{giukie94,belres12,lolpir14}.
They are obtained from an action that consists of the generalized
kinetic part in~\eqref{kin}, i.e, by keeping $\lambda$ (GR is obtained by
setting $\lambda=1$), and the vacuum GR potential in~\eqref{calV},
i.e., a potential arising from $-R$ only, and hence when $k_1=-1$
and $k_2=k_3=k_4=k_5=k_6=0$ in~\eqref{calV}. These models suffice for
our goal of deriving the ODEs \eqref{full:subs}.
The case that modifies GR with more general potentials, the HL models are similar and can be found in~\cite[Appendix A.2]{HellLappicyUggla}. In particular, they heuristically argue that a broad class of HL models possess a dominant potential with asymptotic dynamics described by the $\lambda$-$R$ models.

To obtain succinct expressions for the spatial curvature, and thereby the potential $V={}^1V = - 12gR$,
we introduce the following auxiliary quantities 
\begin{subequations}\label{malpha}
\begin{align}
m_1 &:= n_1g_{11} = n_1 e^{2(2v\beta^\lambda - 2\beta^+)},\\
m_2 &:= n_2g_{22} = n_2 e^{2(2v\beta^\lambda + \beta^+ + \sqrt{3}\beta^-)},\\
m_3 &:= n_3g_{33} = n_3 e^{2(2v\beta^\lambda + \beta^+ - \sqrt{3}\beta^-)}.
\end{align}
\end{subequations}
Here we have introduced the parameter $v$, which is defined by
the relation
\begin{equation}\label{app:v}
v := \frac{1}{\sqrt{2(3\lambda - 1)}},
\end{equation}
and hence $\beta^0 = 2v\beta^\lambda$ due to~\eqref{betaLAMBDAandNcal}.
The parameter $v$ plays a prominent role in the evolution equations. Since we are
interested in continuous deformations of GR with $\lambda=1$, and thus
$v=1/2$, we restrict attention to $v\in (0,1)$.
Specializing the general expression for the spatial curvature in~\cite{elsugg97}
to the diagonal class A Bianchi models leads to 
\begin{equation}\label{R11}
R^1\!_1 = \frac{1}{2g}(m_1^2 - (m_2-m_3)^2),
\end{equation}
where $R^1\!_1 = g^{11}R_{11}
= g_{11}^{-1}R_{11}$, and similarly by permutations
for $R^2\!_2$ and $R^3\!_3$. It follows that the spatial scalar curvature
$R = R^1\!_1 + R^2\!_2 +R^3\!_3$ is given by
\begin{equation}\label{Rscalar}
R = -\frac{1}{2g}(m_1^2 + m_2^2 + m_3^2 - 2m_1m_2 - 2m_2m_3 - 2m_3m_1).
\end{equation}
This thereby yields the potential in~\eqref{VHL} and~\eqref{pots} with $k_1=-1$:
\begin{equation}\label{kinV}
V = {}^1V = -12 g R = 6(m_1^2 + m_2^2 + m_3^2 - 2m_1m_2 - 2m_2m_3 - 2m_3m_1),
\end{equation}
where $V$ depends on $\beta^\lambda$ and $\beta^\pm$ via $m_1$, $m_2$ and $m_3$,
according to equation~\eqref{malpha}.

The evolution equations for $\beta^\lambda$, $\beta^\pm$, $p_\lambda$, $p_\pm$
are obtained from Hamilton's equations, where $T$ and $V$
in the Hamiltonian~\eqref{LambdaRham} are given by~\eqref{Tkinetics} and~\eqref{kinV},
respectively, which yields
\begin{subequations}\label{HamiltonEQ}
\begin{align}
\dot{\beta}^\lambda &= \frac{\partial H}{\partial p_\lambda}
= -{\cal N} p_\lambda, \qquad && \dot{p}_\lambda
= - \frac{\partial H}{\partial\beta^\lambda}
= - \mathcal{N} \frac{\partial V}{\partial\beta^\lambda}, \label{betadotham}\\
\dot{\beta}^\pm &= \frac{\partial H}{\partial{p}_\pm}
= {\cal N}{p}_\pm, \qquad &&\dot{p}_\pm
= - \frac{\partial H}{\partial{\beta}^\pm}
= - \mathcal{N} \frac{\partial V}{\partial{\beta}^\pm},
\end{align}
\end{subequations}
while the Hamiltonian constraint $T+V=0$ is obtained by varying ${\cal N}$.

Next, we choose a new time variable
$\tau_-:=-\beta^\lambda$, 
which is directed toward the physical past, since we are considering
expanding models. This is accomplished by setting ${\cal N} =  p_\lambda^{-1}$
in the first equation in~\eqref{betadotham}, and thereby $N = 12\sqrt{g}/ p_\lambda$,
which results in the following evolution equations:
\begin{subequations}\label{HamiltonEQ2}
\begin{align}
\frac{d\beta^\lambda}{d\tau_-} &= -1, \qquad
&&\frac{d p_\lambda}{d\tau_-} =
-\frac{1}{ p_\lambda} \frac{\partial V}{\partial\beta^\lambda}, \\
\frac{d\beta^\pm}{d\tau_-} &= \frac{{p}_\pm}{ p_\lambda}, \qquad &&
\frac{dp_\pm}{d\tau_-} = -\frac{1}{ p_\lambda} \frac{\partial V}{\partial{\beta}^\pm}.
\end{align}
\end{subequations}

We then rewrite the system~\eqref{HamiltonEQ2} and the constraint $T+V=0$
using the non-canonical variable transformation,
\begin{equation}\label{SigmaNvariables}
\Sigma_\pm := - \frac{p_\pm}{ p_\lambda}, \qquad \qquad \qquad N_\alpha
:= - 2\sqrt{3}\left(\frac{m_\alpha}{ p_\lambda}\right),
\end{equation}
while keeping $p_\lambda$. Note that $\Sigma_\pm = d\beta^\pm/d\beta^\lambda = -d\beta^\pm/d\tau_-$.

These variables lead to a decoupling 
of the evolution equation for the variable $ p_\lambda$,
\begin{equation}\label{p0prime}
p_\lambda^\prime = -4v(1-\Sigma^2) p_\lambda,
\end{equation}
where ${}^\prime$ denotes the derivative $d/d\tau_-$.
This yields the following reduced system of evolution equations
\begin{subequations}\label{dynsyslambdaR}
\begin{align}
\Sigma_\pm^\prime &= 4v(1-\Sigma^2)\Sigma_\pm + {\cal S}_\pm,\\
N_1^\prime &= -2(2v\Sigma^2 - 2\Sigma_+)N_1,\\
N_2^\prime &= -2(2v\Sigma^2 + \Sigma_+ + \sqrt{3}\Sigma_-)N_2,\\
N_3^\prime &= -2(2v\Sigma^2 + \Sigma_+ - \sqrt{3}\Sigma_-)N_3,
\end{align}
while the Hamiltonian constraint $T+V=0$ results in
\begin{equation}\label{constrpmVIIIIX}
1 - \Sigma^2 - \Omega_k =0,
\end{equation}
\end{subequations}
where
\begin{subequations}\label{LambdaRquantities}
\begin{align}
\Sigma^2 &:= \Sigma_+^2 + \Sigma_-^2,\\
\Omega_k &:= N_1^2 + N_2^2 + N_3^2 - 2N_1N_2 - 2N_2N_3 - 2N_3N_1,\\
{\cal S}_+ &:= 2[(N_2 - N_3)^2 - N_1(2N_1 - N_2 - N_3)],\\
{\cal S}_- &:= 2\sqrt{3}(N_2 - N_3)(N_2 + N_3 - N_1).
\end{align}
\end{subequations}

Note that the variables $\Sigma_\pm$, $N_1$, $N_2$ and $N_3$,
defined in~\eqref{SigmaNvariables}, are \emph{dimensionless}. Dimensions
can be introduced in various ways, but terms in a sum
must all have the same dimension. The constraint~\eqref{constrpmVIIIIX}
is such a sum. Since this sum contains 1,
which 
obviously is dimensionless, it follows that $\Sigma_+$, $\Sigma_-$,
$N_1$, $N_2$ and $N_3$ are dimensionless, and so is the time variable
$\tau_-$, as follows from inspection of~\eqref{dynsyslambdaR}.
The vacuum GR equations are obtained by setting $v=1/2$.

In~\cite[Appendix A.2]{HellLappicyUggla}, it is heuristically argued that a broad range of HL models have asymptotic dynamics described by the $\lambda$-$R$ evolution equations \eqref{dynsyslambdaR}. To achieve this, they use Misner's approximation scheme of a `particle' moving in a potential well in $(\beta^+,\beta^-)\in\mathbb{R}^2$ space as $\tau_- = - \beta^\lambda \rightarrow \infty$, which was introduced to understand the initial Bianchi type IX singularity in GR, see~\cite{mis69a,mis69b,waiell97,jan01}.
For HL, each potential term in \eqref{calV} has its associated `moving walls' that move with velocity ${}^iv$. Among those, there is a dominant potential term which yields the same evolution equations as the $\lambda$-$R$ models in \eqref{dynsyslambdaR}, but with different parameters ${}^iv$ given by \eqref{v's} instead of the parameter $v$ in \eqref{app:v}.
}

\textbf{Acknowledgments.} We are grateful for the nonlinear dynamics group at Freie Universität Berlin for insightful discussions/suggestions; in particular, B. Fiedler and H. Sprink. PL was funded by FAPESP, 17/07882-0 and 18/18703-1. VHD was funded by FAPESP, 20/07532-1.

\bibliographystyle{plain}

\begin{thebibliography}{90}

\bibitem{ArJuSa17} A.~Arbieto, A.~Junqueira and B.~Santiago.
\newblock On Weakly Hyperbolic Iterated Function Systems.
\newblock {\it Bull.\ Braz.\ Math.\ Soc.} {\bf 48}, 111-140, (2017).

{\color{black}
\bibitem{Bakas10} I.~Bakas, F.~Bourliot, D.~Lüst and M.~Petropoulos.
\newblock The mixmaster universe in Hořava–Lifshitz gravity.
\newblock {\it Class. Quantum Grav. } {\bf 27}, 045013, (2010).}

\bibitem{Beguin} F.~Beguin.
\newblock Aperiodic oscillatory asymptotic behavior for some Bianchi spacetimes.
\newblock {\it Class.\ Quant.\ Grav.} \textbf{27}, 185005, (2010).

\bibitem{bkl70} V.~A.~Belinski\v{\i}, I.~M.~Khalatnikov, and E.~M.~Lifshitz.
\newblock Oscillatory approach to a singular point
in the relativistic cosmology.
\newblock {\it Adv.\ Phys.} {\bf 19}, 525, (1970).

\bibitem{bkl82} V.~A.~Belinski\v{\i}, I.~M.~Khalatnikov, and E.~M.~Lifshitz.
\newblock A general solution of the Einstein equations with a time singularity.
\newblock {\it Adv.\ Phys.} {\bf 31}, 639, (1982).

{
\color{black}
\bibitem{belres12} J.~Bellor{i}n and A.~Restuccia.
\newblock On the consistency of the Ho\v{r}ava theory.
\newblock {\it Int.\ J.\ Mod.\ Phys.\ D} {\bf 21}, 1250029, (2012).
}

\bibitem{Bernhard}B.~Brehm.
\newblock Bianchi VIII and IX vacuum cosmologies: Almost every solution forms particle horizons and converges to the Mixmaster attractor. \href{https://arxiv.org/abs/1606.08058}{
https://doi.org/10.48550/arXiv.1606.08058}, (2016).

{
\color{black}
\bibitem{LappicyLessard} K. E. M. Church, O. Hénot, P. Lappicy, J.-P. Lessard and H. Sprink.
Periodic orbits in Hořava-Lifshitz cosmologies, \href{https://arxiv.org/abs/2203.03763}{https://doi.org/10.48550/arXiv.2203.03763}, (2022).
}

{
\color{black}
\bibitem{Chernoff83} D. F.~Chernoff and J. D.~Barrow.
\newblock Chaos in the Mixmaster Universe.
\newblock {\it Phys. Rev. Lett.} {\bf 50}, 134, (1983).
}

{
\color{black}
\bibitem{Cornish97} N. J.~Cornish and J. J.~Levin.
\newblock Mixmaster universe: A chaotic Farey tale.
\newblock {\it Phys. Rev. D} {\bf 55}, 7489, (1997).
}

\bibitem{Devaney} R.~Devaney.
\newblock {\it An introduction to chaotic dynamical systems.}
\newblock  Avalon Publishing, 2nd Ed., (1989).

\bibitem{Dutilleul}
T. Dutilleul. Chaotic dynamics of spatially homogeneous spacetimes. Phd Thesis, Université Paris 13 - Sorbonne Paris Cité, (2019).
\href{https://tel.archives-ouvertes.fr/tel-02488655/file/These_Tom_Dutilleul.pdf}{Link.}

{
\color{black}
\bibitem{elsugg97} H.~van~Elst and C.~Uggla.
\newblock General relativistic 1+3 orthonormal frame approach revisited
\newblock {\it Class.\ Quant.\ Grav.} {\bf 14}, 2673, (1997).
}

{
\color{black}
\bibitem{GiKam17} L.~Giani and A. Y.~Kamenshchik.
\newblock Hořava–Lifshitz gravity inspired Bianchi-II cosmology and the mixmaster universe.
\newblock {\it Class.\ Quant.\ Grav.} {\bf 34}, 085007, (2017).
}

{
\color{black}
\bibitem{giukie94} D.~Giulini and C.~Kiefer.
\newblock Wheeler-DeWitt metric and the attractivity of gravity.
\newblock {\it Phys.\ Lett.\ A} {\bf 193}, 21, (1994).
}

\bibitem{heiugg09b} J.~M.~Heinzle and C.~Uggla.
\newblock A new proof of the Bianchi type IX attractor theorem.
\newblock {\it Class.\ Quant.\ Grav.} {\bf 26}, 075015, (2009).

\bibitem{Mixmaster}
J.~M.~Heinzle and C.~Uggla.
\newblock Mixmaster: Fact and Belief.
\newblock {\it Class.\ Quant.\ Grav.} \textbf{26}, 075016, (2009).

\bibitem{HellLappicyUggla} J.~Hell, P.~Lappicy and C.~Uggla.
Bifurcations and Chaos in Ho{\v{r}}ava-Lifshitz Cosmology, \href{https://arxiv.org/abs/2012.07614}{https://doi.org/10.48550/arXiv.2012.07614}, (2020).

{
\color{black}
\bibitem{Hobill94}
D.~Hobill, A.~Burd and A.~Coley (eds).
\newblock {\it Deterministic chaos in general relativity.} 
\newblock Plenum Press, New York (1994).
}

\bibitem{hor09a} P.~Ho\v{r}ava.
\newblock Membranes at Quantum Criticality.
\newblock {\it J.\ High Energy Phys.} {\bf 0903}, 020, (2009).

\bibitem{hor09b} P.~Ho\v{r}ava.
\newblock Quantum Gravity at a Lifshitz Point.
\newblock {\it Phys.\ Rev.\ D} {\bf 79}, 084008, (2009).

\bibitem{Hutch81} J.~Hutchinson.
\newblock Fractals and Self Similarity.
\newblock Indiana\ Univ.\ Math. \ J. \textbf{30}, pp. 713-747, (1981).

{
\color{black}
\bibitem{jan01} R.T. Jantzen.
\newblock {\em Spatially Homogeneous Dynamics: A Unified
Picture}.
\newblock in {\it Proc.\ Int.\ Sch.\ Phys.\ ``E. Fermi" Course LXXXVI on
``Gamov Cosmology"\/}, R. Ruffini, F. Melchiorri, Eds. North Holland,
Amsterdam, (1987) and in  {\it Cosmology of the Early Universe\/}, R. Ruffini,
L.Z. Fang, Eds., World Scientific, Singapore, (1984).
}


\bibitem{khaetal85} I.~M.~Khalatnikov, E.~M.~Lifshitz, K.~M.~Khanin, L.~N.~Shur, and Y.~G.~Sinai.
\newblock On the stochasticity in relativistic cosmology.
\newblock {\it J.\ Stat.\ Phys.} {\bf 38}, 97, (1985). 

\bibitem{Liebscher}
S.~Liebscher, J.~Harterich, K.~Webster and M.~Georgi.
\newblock Ancient dynamics in Bianchi models: Approach to Periodic Cycles.
\newblock
{\it Commun.\ Math.\ Phys.} \textbf{305}, 59-83, (2011).

\bibitem{Lifshitz41} E.~Lifshitz.
\newblock On the theory of second-order phase transitions I.
\newblock {\it Zh.\ Eksp. \ Teor. \ Fiz} {\bf 11}, 255, (1941).

\bibitem{Lima19} Y.~Lima.
\newblock Symbolic dynamics for nonuniformly hyperbolic systems.
\newblock {\it Ergodic\ Theory\ Dyn.\ Sys.} {\bf 41}, 2591-2658, (2020). 

{
\color{black}
\bibitem{lolpir14} R.~Loll and L.~Pires.
\newblock Role of the extra coupling in the kinetic term in Ho\v{r}ava-Lifshitz gravity.
\newblock {\it Phys.\ Rev.\ D} {\bf 90}, 124050, (2014).
}


\bibitem{MatDiaz} E.~Matias and L.J.~Diaz.
\newblock Non-hyperbolic Iterated Function Systems: semifractals and the chaos game.
\newblock {\it Fund.\ Math.} {\bf 250}, 21-39, (2020).

{
\color{black}
\bibitem{mis69a} C.~W.~Misner.
\newblock Mixmaster universe.
\newblock {\it Phys.\ Rev.\ Lett.} {\bf 22}, 1071, (1969).

\bibitem{mis69b} C.~W.~Misner.
\newblock Quantum cosmology I.
\newblock {\it Phys.\ Rev.} {\bf 186}, 1319, (1969).

\bibitem{grav73}
C.~W.~Misner, K~.S.~Thorne, and J.~A.~Wheeler.
\newblock {\em Gravitation}.
\newblock W.~H.~Freeman and Company, San Francisco, (1973).

\color{black}
\bibitem{Miso11} Y.~Misonoh, K.~Maeda, and T.~Kobayashi.
\newblock Oscillating Bianchi IX universe in Hořava-Lifshitz gravity.
\newblock {\it Phys. Rev. D} {\bf 84}, 064030, (2011).
}

\bibitem{Ringstrom} H.~Ringström.
\newblock The Bianchi IX attractor.
\newblock {\it Annales\ Henri\ Poincaré} \textbf{2}, 405-500 (2001).

\bibitem{HL_status_report} T.~P.~Sotiriou.
\newblock Horava-Lifshitz gravity: a status report.
\newblock{\it J. Phys. Conf. Ser.} {\bf 283}, 012034, (2011).

\bibitem{ugg13a} C. Uggla.
\newblock Recent developments concerning generic spacelike singularities.
\newblock {\it Gen.\ Rel.\ Grav.} {\bf 45}, 1669, (2013).

\bibitem{ugg13b} C.~Uggla.
\newblock Spacetime Singularities: Recent Developments.
\newblock {\it Int.\ J.\ Mod.\ Phys.\ D} {\bf 22}, 1330002, (2013).

{
\color{black}
\bibitem{waiell97} J. Wainwright and G.F.R. Ellis.
\newblock {\em Dynamical systems in cosmology}.
\newblock Cambridge University Press, Cambridge, (1997).
}

\bibitem{DeWitt67} B.~DeWitt.
\newblock Quantum Theory of Gravity. I. The Canonical Theory.
\newblock {\it Phys.\ Rev.} {\bf 160}, 1113, (1967).




\end{thebibliography}

\end{document}